\documentclass{article}

\usepackage{amsmath,amssymb,amsthm,amsfonts}
\usepackage{algorithmic,algorithm}
\usepackage{color,braket}

\renewcommand{\P}{\mathcal{P}}
\newcommand{\PQ}{\mathcal{P}_Q}
\newcommand{\PQp}{\mathcal{P}_{Q^\perp}}
\newcommand{\PT}{\mathcal{P}_T}
\newcommand{\PTp}{\mathcal{P}_{T^\perp}}
\newcommand{\PO}{\mathcal{P}_\Omega}
\newcommand{\POij}{\mathcal{P}_{\Omega_{ij}}}
\newcommand{\POp}{\mathcal{P}_{\Omega^\perp}}
\newcommand{\POj}{\mathcal{P}_{\Omega_j}}

\newcommand{\PG}{\mathcal{P}_\Gamma}
\newcommand{\PGp}{\mathcal{P}_{\Gamma^\perp}}

\newcommand{\PP}{\mathcal{P}_\Pi}
\newcommand{\PPp}{\mathcal{P}_{\Pi^\perp}}

\newcommand{\PR}{\mathcal{P}_R}

\newcommand{\R}{\mathbb{R}}
\newcommand{\sgn}{\mathrm{sgn}}
\newcommand{\ber}{\mathrm{Ber}}
\newcommand{\rank}{\mathrm{rank}}
\newcommand{\trace}{\mathrm{trace}}
\newcommand{\supp}{\mathrm{supp}}
\renewcommand{\Pr}{\mathbb{P}}
\newcommand{\E}{\mathbb{E}}
\newcommand{\Qp}{{Q^\perp}}
\renewcommand{\vec}{\mathrm{vec}}

\renewcommand{\H}{\mathcal{H}}
\newcommand{\N}{\mathcal{N}}
\newcommand{\bz}{\mathbf{0}}
\newcommand{\st}{\mathrm{subj. \: to}}
\newcommand{\be}{\bar{e}}

\newcommand{\x}{\mathbf{x}}
\newcommand{\y}{\mathbf{y}}
\newcommand{\id}{\mathcal{I}}

\theoremstyle{definition}
\newtheorem{thm}{Theorem}

\newtheorem{remark}{Remark}

\theoremstyle{definition}
\newtheorem{lemma}{Lemma}

\theoremstyle{definition}

\theoremstyle{definition}
\newtheorem{fact}{Fact}

\theoremstyle{definition}
\newtheorem{cor}{Corollary}

\title{Principal Component Pursuit with Reduced Linear Measurements}
\author{Arvind Ganesh$^*$, Kerui Min$^*$, John Wright$^\dagger$, and Yi Ma$^{*,\ddagger}$\\$^*$ {\small Dept. of Electrical and Computer Engineering, University of Illinois at Urbana-Champaign}\\
$^\dagger$ {\small Dept. of Electrical Engineering, Columbia University} \\
$^\ddagger$ {\small Visual Computing Group, Microsoft Research Asia}
}
\date{}

\begin{document}
\maketitle

\begin{abstract}
In this paper, we study the problem of decomposing a superposition of a low-rank matrix and a sparse matrix when a relatively few linear measurements are available. This problem arises in many data processing tasks such as aligning multiple images or rectifying regular texture, where the goal is to recover a low-rank matrix with a large fraction of corrupted entries in the presence of nonlinear domain transformation. We consider a natural convex heuristic to this problem which is a variant to the recently proposed Principal Component Pursuit. We prove that under suitable conditions, this convex program guarantees to recover the correct low-rank and sparse components despite reduced measurements. Our analysis covers both random and deterministic measurement models.

\end{abstract}

\section{Introduction}
\label{sec:intro}

Low-rank matrix recovery and approximation has been a popular area of research in many different fields. The popularity of low-rank matrices can be attributed to the fact that they arise in one of the most commonly used data models in real applications, namely when very high-dimensional data samples are assumed to lie approximately on a low-dimensional linear subspace. This model has been successfully employed in various problems such as face recognition \cite{Wright2009-PAMI}, system identification \cite{Fazel2004-ACC}, and information retrieval \cite{Papadimitriou2000-JCSS}, for instance.

The most popular tool for low-rank matrix approximation is the Principal Component Analysis (PCA) \cite{Eckart1936-Psychometrika,Jolliffe1986}. The basic idea of PCA is to find the ``best low-rank approximation'' (in an $\ell_2$-sense) to a given input matrix. Essentially, PCA finds a rank-$r$ approximation to a given data matrix $D \in \R^{m \times n}$ by solving the following problem:
$$
\min_L \, \|D-L\| \quad \mathrm{s.t.} \quad \rank(L) \leq r,
$$
where $\|\cdot\|$ denotes the matrix spectral norm. It is well-known that the solution to this problem can be easily obtained by computing the Singular Value Decomposition (SVD) of $D$ and retaining only the $r$ largest singular values and the corresponding singular vectors. Besides the ease of computation, the PCA estimate has been shown to be optimal in the presence of isotropic Gaussian noise. However, the biggest drawback of PCA is that it breaks down even when one entry of the matrix is corrupted by an error of very large magnitude. Unfortunately, such large-magnitude, {\it non-Gaussian} errors often exist in real data. For instance, occlusions in images corrupt only a fraction of the pixels in an image, but the magnitude of corruption can be quite large.

There have been many works in the literature that try to make PCA robust to such gross, non-Gaussian errors and many models and solutions have been proposed. We here consider the specific problem of recovering a low-rank matrix $L_0 \in \R^{m\times n}$ from corrupted observations $D = L_0 + S_0$, where $S_0 \in \R^{m \times n}$ is a sparse matrix whose non-zero entries may have arbitrary magnitude. This problem has been studied in detail recently by various works in the literature \cite{Candes2011-JACM,Chandrasekaran2011-SJO,Hsu2011-IT}. It has been shown that under rather broad conditions, the following convex program succeeds in recovering $L_0$ from $D$:
\begin{equation}
\min_{L,S} \, \|L\|_* + \lambda \|S\|_1 \quad \mathrm{s.t.} \quad D = L + S,
\label{eqn:pcp}
\end{equation}
where $\|\cdot\|_*$ denotes the nuclear norm\footnote{The sum of all singular values.}, $\|\cdot\|_1$ denotes the $\ell_1$-norm\footnote{The sum of absolute values of all matrix entries.}, and $\lambda > 0$ is a weighting factor. This method has been dubbed Principal Component Pursuit (PCP) in \cite{Candes2011-JACM}. In addition to being computationally tractable, it comes with very strong theoretical guarantees of recovery. Furthermore, follow-up works have shown that PCP is stable in the presence of additive Gaussian noise \cite{Zhou2010-ISIT} and can recover $L_0$ even when the corruption matrix $S_0$ is not so sparse \cite{Ganesh2010-ISIT}.

Besides being of theoretical interest, this convex optimization framework for low-rank matrix recovery has been employed very successfully to solve real problems in computer vision such as photometric stereo \cite{Wu2010-ACCV}. However, in practice, much more data, especially imagery data, can be viewed as low-rank only after some transformation is applied.  For instance, an image of a building facade will become a low-rank matrix after the perspective distortion is rectified \cite{Zhang2011-IJCV} or a set of face images of the same person will become linearly correlated only after they are proper aligned  \cite{Peng2011-PAMI}. With our terminology here, we can write as $D \circ \tau = L_0 + S_0$ where $\tau$ belongs to certain transformation group. As the transformation $\tau$ is also unknown, one natural way to recover $L_0, S_0$ and $\tau$ together is to approximate the nonlinear equation with its linearization at the current estimate of $\hat{\tau}$: $$D \circ \hat{\tau}  + \sum_{i=1}^p J_i d\tau_i = L + S,$$ where $\{ J_i \} $ is the Jacobian of $D \circ \tau$ with respect to the parameters $\{ \tau_i \}$ of $\tau$. Then one can incrementally update the estimate for $\tau$ with $\hat{\tau} + d\tau$ by solving the following convex program:
\begin{equation}
\min_{L,S, d\tau_i } \, \|L\|_* + \lambda \|S\|_1 \quad \mathrm{s.t.} \quad D + \sum_{i=1}^p J_i d\tau_i = L + S.
\label{eqn:rasl}
\end{equation}
Empirically this scheme has been shown to work rather effectively in practice in both the image rectification problem \cite{Zhang2011-IJCV} and the image alignment problem \cite{Peng2011-PAMI}.

Although the convex program was proposed in the same spirit as PCP, we note that the linear constraint is different, and hence, the theoretical guarantees for PCP shown in \cite{Candes2011-JACM,Chandrasekaran2011-SJO,Hsu2011-IT} do not directly apply to this case. In this work, we attempt to fill the gap between theory and practice and try to understand under what conditions, the above extended version of PCP is expected to work correctly.

Let $Q$ be the linear subspace in $\R^{m \times n}$ that is the orthogonal complement to the span of all the $J_i$'s, then its dimension is $q = mn - p$. Clearly, we can rewrite the above program in the following form:
\begin{equation}
\min_{L,S} \, \|L\|_* + \lambda \|S\|_1 \quad \mathrm{s.t.} \quad \PQ D = \PQ(L + S),
\label{eqn:cpcp}
\end{equation}
where $\PQ$ is the orthogonal projection onto the linear subspace $Q$. Clearly, this program is a variation to PCP \eqref{eqn:pcp} in which the number of linear constraints has been reduced from $mn$ to $q = mn - p$. Indeed, if $Q$ is the entire space, then it reduces to the PCP. If $Q$ is a linear subspace of matrices with support in $\Omega \subseteq [m] \times [n]$, then we have the special case of recovering $L_0$ from $D$, when only a subset of the entries in $D$ are available. This case is akin to the low-rank matrix completion problem \cite{Candes2008,Candes2009,Gross2009-pp}, and theoretical guarantees have been derived in \cite{Candes2011-JACM,Li2011-pp}. However, to the best of our knowledge, the case with a general subspace $Q$ has not yet been analyzed in detail in the literature.

Our motivation to study when the convex program \eqref{eqn:cpcp} succeeds with such reduced linear constraints is at least twofold. First, the relationships between $Q$ and $L_0$ and $S_0$ will provide us better understanding about what type of images and signals for which techniques such as those used in \cite{Zhang2011-IJCV, Peng2011-PAMI} are expected to work well.  Second, we want to know how many general linear measurements we could reduce without sacrificing the robustness of PCP for recovering the low-rank matrix $L_0$. In these applications, the number of constraints reduced corresponds to the dimension of the transformation group. In the image rectification problem, the dimension of the transformation group $p$ is typically fixed with respect to the size of the matrix; in the image alignment problem, however, the dimension typically grows linearly in $m$ (or $n$). In either case, we need to know if the program \eqref{eqn:cpcp} can tolerate up to a constant fraction of gross errors.


\subsection{Notation}

We first establish a set of notations that will be used throughout this work. We will assume that the matrices $L_0$, $S_0$ and $D$ in \eqref{eqn:cpcp} have size $m \times n$. Without any loss of generality, we assume that $n \leq m$. We denote the rank of $L_0$ by $r$. Let $L_0 = U\Sigma V^*$ be the reduced Singular Value Decomposition (SVD) of $L_0$. We define a linear subspace $T$ as follows:
\begin{equation}T \doteq \{UX^* + YV^* : X \in \R^{n \times r}, Y \in \R^{m \times r}\}.
\label{eqn:T}
\end{equation}
Basically, $T$ contains all matrices that share a common row space or column space with $L_0$. We denote by $\Omega$ the support of $S_0$. By a slight abuse of notation, we also represent by $\Omega$ the subspace of matrices whose support is contained in the support of $S_0$. For any subspace $S \subseteq \R^{m \times n}$, $\P_S : \R^{m \times n} \rightarrow \R^{m \times n}$ denotes the orthogonal projection operator onto $S$.

For any $X,Y \in \R^{m \times n}$, we define their inner product as $ \langle X, Y\rangle = \trace(X^*Y) = \sum_{ij} X_{ij} Y_{ij}. $ We let $\|\cdot \|_F$ and $\|\cdot\|$ denote the matrix Frobenius norm and spectral norm, respectively. We also denote the $\ell_\infty$-norm of a matrix $X$ as $\|X\|_\infty = \max_{ij} \, |X_{ij}|.$ We say that an event $E$ occurs with high probability if $\Pr[E^c] \leq C\, m^{-\alpha},$ for some positive numerical constants $C$ and $\alpha$. Here, $E^c$ denotes the event complement to $E$.
\medbreak

\subsection{Main Assumptions}
\label{sec:assumptions}

Obviously, successful recovery is not always guaranteed except under proper assumptions on the low-rank $L_0$, sparse $S_0$, and the subspace $Q$ involved. For instance, if the matrix $L_0$ is itself a sparse matrix, then there is a fundamental ambiguity in the solution to be recovered. Here, we outline some of our assumptions that we will use throughout this paper. The assumptions we make here on $L_0$ and $S_0$ are essentially the same as those for PCP \cite{Candes2011-JACM}. For completeness, we list them below.

We assume that each entry of the matrix belongs to the support of the sparse matrix $S_0$ independently with probability $\rho$. We denote this as $\supp(S_0) \sim \ber(\rho)$. For simplicity, we assume the signs of the nonzero entries are also random.\footnote{The random sign assumption is not entirely necessary for obtaining the same qualitative results. One can follow the derandomization process in \cite{Candes2011-JACM} to remove this assumption if needed.} For the low-rank matrix $L_0$, we assume the subspace $T$ defined in\eqref{eqn:T} is incoherent to the standard basis (and hence the sparse matrix $S_0$). To be precise, let us denote the standard basis in $\R^m$ and $\R^n$ by $\bar{e}_i$ and $e_j$, respectively, where $i \in [m]$ and $j \in [n]$. We assume (as in \cite{Candes2008}) that
\begin{eqnarray}
\max_{i \in [m]} \|U^*\bar{e}_i\|_2^2 \leq \frac{\mu r}{m}, \quad
 \max_{j \in [n]} \|V^*e_j\|_2^2 \leq \frac{\mu r}{n}, \quad
 \|UV^*\|_\infty \leq \sqrt{\frac{\mu r}{mn}},
\label{eq:assumptionUV}
 \label{eq:assumptionVe}
 \label{eq:assumptionUe}
 \label{eqn:incoh_subspace}
\end{eqnarray}
for some $\mu > 0$ and for all $(i,j) \in [m] \times [n]$. We recall that $r = \rank(L_0)$. It follows from the above assumptions that for any $(i,j) \in [m] \times [n]$
\begin{equation}
\|\PT \bar{e}_i e_j^*\|_F \leq \sqrt{\frac{2 \mu r}{n}}.
\label{eq:PTeiej}
\end{equation}
Furthermore, it can be shown that $\|\PTp X\| \leq \|X\|$ for any $X \in \R^{m \times n}$.

In addition to the above assumptions, we define the following two properties of linear subspaces. We say that a linear subspace $S \subseteq \R^{m \times n}$ is
\begin{itemize}
\item \emph{$\nu$-coherent} if there exists an orthonormal basis $\{G_i\}$ for $S$ satisfying
\begin{equation}
\max_i \|G_i\|^2 \le \frac{\nu}{n}.
\label{eqn:nu-coherence}
\end{equation}
\item \emph{$\gamma$-constrained} if
\begin{equation}
\max_{i,j} \|\P_S\bar{e}_i e_j^*\|_F^2 \le \gamma.
\label{eqn:gamma-constrained}
\end{equation}
\end{itemize}
If $S$ is a random subspace, we say it is $\nu$-coherent and $\gamma$-constrained if Eqns. \eqref{eqn:nu-coherence} and \eqref{eqn:gamma-constrained} hold with high probability, respectively.

In this paper, we will deal with two different assumptions on the subspace $Q$ as outlined below. We will see later that it is in fact convenient to make our assumptions on the subspace $Q^\perp$, rather than on $Q$ itself. This is partly motivated from the model in \eqref{eqn:rasl} that was used in \cite{Zhang2011-IJCV, Peng2011-PAMI}, where the $J_i$'s are essentially a basis for $Q^\perp$. So, any assumptions on $Q^\perp$ can be easily interpreted in terms of the $J_i$'s and this would help us make the connection to these applications more directly. We denote by $p$ the dimension of the subspace $Q^\perp$.
\begin{itemize}
\item {\bf Random subspace model.} Let $G_1, G_2, \ldots, G_p \in \R^{m \times n}$ be an orthonormal basis for $Q^\perp$. We assume that this basis set is chosen uniformly at random from all possible orthobasis sets of size $p$ in $\R^{m \times n}$. It can be shown that each of the $G_i$'s are identical in distribution to $H / \|H\|_F$, where the entries of $H\in \R^{m \times n}$ are i.i.d. according to a Gaussian distribution with mean 0 and variance $1/mn$.
\item {\bf Deterministic subspace model.} Under this model, we assume that $\Qp$ is a fixed subspace which is $\nu$-coherent, for some $\nu\geq 1$.
\end{itemize}

\subsection{Main Results}
With the above notation, we now briefly describe the main results we prove in this work. Although our results and proof methodology resemble those in \cite{Candes2011-JACM}, there are some important differences here. Particularly, we will see that the assumptions we make on the subspace $Q$ greatly influences the kind of guarantees for recovery that can be derived.

As mentioned earlier, we will consider two different assumptions on the subspace $Q$. In the first one, we assume a {\it random subspace model} for $Q^\perp$. The main result that we prove in this work under this random subspace model is summarized as the following theorem.




\begin{thm}[\bf Random Reduction] Fix any $C_p > 0$, and let $Q^\perp$ be a $p$-dimensional random subspace of $\R^{m \times n}$ ($n \leq m$), $L_0$ a rank-$r$, $\mu$-incoherent matrix, and $\supp(S_0) \sim \ber(\rho)$. Then, provided that
\begin{equation}
r < C_r \frac{n}{\mu \log^2 m}, \quad p < C_p n, \quad \rho < \rho_0,
\end{equation}
with high probability $(L_0, S_0)$ is the unique optimal solution to \eqref{eqn:cpcp} with $\lambda = m^{-1/2}$. Here, $C_r > 0$ and $\rho_0 \in (0,1)$ are numerical constants.
\label{thm:cpcp_random}
\end{thm}

\begin{remark}
In Theorem \ref{thm:cpcp_random}, ``with high probability'' means with probability at least $1- \beta(C_p) m^{-c}$, with $c > 0$ numerical.
\end{remark}

The scaling in this result covers several applications of interest: in \cite{Zhang2011-IJCV}, $p$ is a fixed constant, while in \cite{Peng2011-PAMI}, $p$ scales linearly with $n$. Therefore, the above result already covers both these applications in terms of the number of reduced constraints. It states that with such reduced constraints, the convex program \eqref{eqn:cpcp} can recover the low-rank matrix $L_0$ essentially under the same conditions as PCP. In particular, it can tolerate up to a constant fraction of errors.

In a work that is closely related to this one \cite{Wright}, we have shown that one can expect the convex program \eqref{eqn:cpcp} to work under much more highly compressive scenario. More precisely, the dimension of the subspace $Q$ only needs to be on the order of $(mr + k) \log^2m$ which is only a polylogarithmic factor more than the intrinsic degrees of freedom of the unknown $L_0$ and $S_0$. One nice feature about the work of \cite{Wright} is that the proof framework is very modular and the techniques are even applicable to more general structured signals beyond low-rank and sparse ones. Nevertheless, that result does not subsume the result here because in such highly compressive scenario, we cannot expect to tolerate error up to a constant fraction of the matrix entries. Obtaining the results in Theorem \ref{thm:cpcp_random} and Theorem \ref{thm:cpcp_deterministic} seems to require arguments that are specially tailored to the PCP problem. 

There is a common limitation for all results that are based on a random assumption for $Q$ or $Q^\perp$: the random assumption does not hold in many real applications. For instance, in \cite{Zhang2011-IJCV,Peng2011-PAMI}, the subspace $Q^\perp$ is typically spanned by a set of image Jacobians, which may not behave like random matrices.  Therefore, it is desirable to have deterministic conditions on $Q^\perp$ (or $Q$) that can be verified for the given data. We need theoretical guarantees for recovery when $Q^\perp$ is a deterministic subspace. This is the second scenario that we will consider in this work, for which we have the following result:


\begin{thm}[\bf Deterministic Reduction]
Fix any $p \in \mathbb{Z}_+$, $\alpha \ge 1$, and $\nu \ge 1$. Then there exists $C_r> 0$ such that if $Q^\perp$ is a $\nu$-coherent $p$-dimensional subspace of $\R^{m \times n}$ ($n\leq m \leq \alpha n$), $L_0$ is a rank-$r$, $\mu$-incoherent matrix, and $\mathrm{supp}(S_0) \sim \mathrm{Ber}(\rho)$, with high probability $(L_0, S_0)$ is the unique optimal solution to \eqref{eqn:cpcp} with $\lambda = m^{-1/2}$, provided that
\begin{equation}
r < C_r \min\left\{\left(\frac{n}{\nu^2 p^2\alpha}\right)^{1/2}, \left(\frac{n}{\alpha \nu \mu p}\right)^{1/3}, \frac{n}{\mu \log m}\right\}, \quad \rho < \rho_0,
\end{equation}
where $C_r, \rho_0 \in (0,1)$ are numerical constants.
\label{thm:cpcp_deterministic}
\end{thm}
\begin{remark}
Here, ``with high probability'' means with probability at least $1- \beta(p,\alpha,\nu) m^{-c}$, with $c > 0$ numerical.
\end{remark}

The $\nu$-coherence condition essentially requires there exists an orthonormal basis for $Q^\perp$ whose spectral norms are bounded above by $O(n^{-1/2})$. This is a condition that can be verified directly once the subspace $Q$ or $Q^\perp$ is given (say as the span of the Jacobians). This condition is also significantly weaker than the random subspace assumption in Theorem \ref{thm:cpcp_random}.

Because the assumptions are weaker, the orders of growth in Theorem \ref{thm:cpcp_deterministic}, quite a bit more restrictive than those in Theorem \ref{thm:cpcp_random}. Nevertheless, this result can be very useful for the practical problems that we encountered in image rectification where the dimension of the transformation group is typically fixed (i.e.\ does not change with the matrix dimension). Theorem \ref{thm:cpcp_deterministic} suggests we should expect the program to work at least for deformation groups whose dimension is fixed. Although empirical results suggest that it could even grow as $O(n)$, we leave that for future investigation.

\paragraph{\bf The remainder of this paper is organized as follows:} In Section \ref{sec:dual_cert}, we derive the optimality conditions for $(L_0,S_0)$ to be the optimal solution to the convex program \eqref{eqn:cpcp}. In particular, we derive the conditions that a certain dual certificate must satisfy that would establish our main result. In Section \ref{sec:dual_cons}, we provide a constructive procedure for the aforementioned dual certificate. In Section \ref{sec:main_random_proof}, we describe our main assumptions and the detailed steps of the proof of Theorem \ref{thm:cpcp_random}. In Section \ref{sec:deterministic_proof}, we outline the proof of Theorem \ref{thm:cpcp_deterministic}. Although the proof for both the deterministic case will follow a common strategy as the random case, there are a few important differences. In particular, we will highlight the parts where the proof deviates significantly from that of Theorem \ref{thm:cpcp_random}.

\section{Existence of Dual Certificate}
\label{sec:dual_cert}

In this section, we prove the following lemma that establishes necessary and sufficient conditions for $(L_0,S_0)$ to be the optimal solution to \eqref{eqn:cpcp}.

\begin{lemma}
Assume that $\dim(Q^\perp \oplus T \oplus \Omega) = \dim(Q^\perp)+\dim(T)+\dim(\Omega)$. $(L_0,S_0)$ is the unique optimal solution to \eqref{eqn:cpcp} if there exists a pair $(W,F) \in \R^{m \times n}\times \R^{m \times n}$ satisfying
\begin{equation}
UV^* + W = \lambda (\sgn(S_0)+F) \in Q,
\label{eqn:dual}
\end{equation}
with $\PT W = \bz, \|W\| < 1, \PO F = \bz$, and $\|F\|_\infty < 1$.
\label{lem:dual}
\end{lemma}
\begin{proof}
Consider a feasible solution to \eqref{eqn:cpcp} of the form $(L_0+H_L,S_0-H_S)$. Clearly, we have that $\PQ H_L = \PQ H_S$. Under the conditions mentioned in the lemma, we will show that this pair does not minimize the cost function in \eqref{eqn:cpcp}, unless $H_L = H_S = \bz$.

We first use the fact that $\|\cdot\|_*$ and $\|\cdot\|_1$ are convex functions. Consider any pair $(W_0,F_0) \in \R^{m \times n}\times \R^{m \times n}$ satisfying $\PT W_0 = \bz$, $\|W_0\| \leq 1$, $\PO F_0 = \bz$, and $\|F_0\|_\infty \leq 1$. Then, $UV^* + W_0$ is a subgradient to $\|\cdot\|_*$ at $L_0$, and $\sgn(S_0) + F_0$ is a subgradient to $\|\cdot\|_1$ at $S_0$. Therefore,
$$
\|L_0+H_L\|_* + \lambda \|S_0-H_S\|_1 \geq \|L_0\|_* + \lambda \|S_0\|_1 + \langle UV^* + W_0, H_L\rangle - \lambda \langle \sgn(S_0) + F_0, H_S\rangle.
$$
By H\"{o}lder's inequality (and the duality of norms), it is possible to choose $W_0$ and $F_0$ such that
$$
\langle W_0, H_L \rangle = \|\PTp H_L\|_*, \quad \langle F_0, H_S\rangle = -\|\POp H_S\|_1.
$$
Then, we have
\begin{eqnarray*}
\|L_0+H_L\|_* + \lambda \|S_0-H_S\|_1& \geq & \|L_0\|_* + \lambda \|S_0\|_1 + \langle UV^*, H_L \rangle - \lambda\langle \sgn(S_0), H_S \rangle \\
& & \quad + \|\PTp H_L\|_* + \lambda \|\POp H_S\|_1.
\end{eqnarray*}
By assumption, we have
$$
UV^* = \lambda (\sgn(S_0) + F) - W,
$$
with $ \lambda (\sgn(S_0) + F) \in Q$. Substituting for $UV^*$ and using $\PQ H_L = \PQ H_S$, we get
$$
\langle UV^*, H_L\rangle = \lambda \langle \sgn(S_0), H_S \rangle  + \lambda \langle F, H_S\rangle - \langle W, H_L \rangle.
$$
Substituting this in the above inequality, we get
\begin{eqnarray*}
\|L_0+H_L\|_* + \lambda \|S_0-H_S\|_1 &\geq& \|L_0\|_* + \lambda \|S_0\|_1 + \|\PTp H_L\|_* + \lambda \|\POp H_S\|_1 \\
& & \quad+ \lambda \langle F, H_S \rangle - \langle W, H_L \rangle.
\end{eqnarray*}
Let $\beta = \max\{\|W\|,\|F\|_\infty\} < 1$. Using H\"{o}lder's inequality, we get
\begin{eqnarray*}
\|L_0+H_L\|_* + \lambda \|S_0-H_S\|_1 &\geq& \|L_0\|_* + \lambda \|S_0\|_1 + (1-\beta) \|\PTp (H_L)\|_* \\
& & \quad+ (1-\beta) \lambda \|\POp (H_S)\|_1.
\end{eqnarray*}
For non-zero $H_L, H_S$, the last term on the right hand side above can be zero only if $H_L \in T \backslash \{\bz\}$ and $H_S \in \Omega\backslash \{\bz\}$. Since $\Omega \cap T = \{\bz\}$, $H_L \neq H_S$. We also have $\PQ (H_L- H_S) = \bz$. This implies that $H_L - H_S \in Q^\perp$, which is a contradiction since $Q^\perp \cap (T \oplus \Omega) = \{\bz\}$. Thus, we have
$$
\|L_0+H_L\|_* + \lambda \|S_0-H_S\|_1 > \|L_0\|_* + \lambda \|S_0\|_1,
$$
for any non-zero {\it feasible} perturbation $(H_L, H_S)$.
\end{proof}

It is often convenient to relax the equality constraints on the dual certificate given in \eqref{eqn:dual}. Thus, similar to the proof outline in \cite{Candes2011-JACM,Gross2009-pp}, we now provide a slightly relaxed dual certificate condition.

\begin{fact}
Let $S_1$ and $S_2$ be two linear subspaces in $\R^{m \times n}$ with $S_1  \subseteq S_2$. Then, for any $X \in \R^{m \times n}$, we have $\P_{S_1}X = \P_{S_1}\P_{S_2}X$, and consequently, $\|\P_{S_1}X\|_F \leq \|\P_{S_2}X\|_F$.
\label{fact:subproj}
\end{fact}

\begin{lemma}
Suppose that $\dim(Q^\perp \oplus T \oplus \Omega) = \dim(Q^\perp)+\dim(T)+\dim(\Omega)$. Let $\Gamma = Q \,\cap\, T^\perp$ so that $\Gamma^\perp = Q^\perp \oplus T$. Assume that $\|\PO\PGp\| < 1/2$ and $\lambda < 1$. Then, $(L_0,S_0)$ is the unique optimal solution to \eqref{eqn:cpcp} if there exists a pair $(W,F) \in \R^{m \times n}\times \R^{m \times n}$ satisfying
\begin{equation}
UV^* + W = \lambda (\sgn(S_0)+F + \PO D) \in Q,
\label{eqn:reldual2}
\end{equation}
with $\PT W = \bz, \|W\| < 1/2, \PO F = \bz$, $\|F\|_\infty < 1/2$, and $\|\PO D\|_F \leq 1/4$.
\label{lem:reldual2}
\end{lemma}
\begin{proof}
Proceeding along the same lines as in the proof of Lemma \ref{lem:dual}, for any feasible perturbation $(H_L, H_S)$, we get
\begin{eqnarray*}
\|L_0 + H_L\|_* + \lambda \|S_0-H_S\|_1 & \geq & \|L_0\|_* + \lambda \|S_0\|_1 + \frac{1}{2} \|\PTp H_L\|_*  \\
& & \quad + \frac{\lambda}{2} \|\POp H_S\|_1 + \lambda \langle \PO D, H_S\rangle \\
 & \geq & \|L_0\|_* + \lambda \|S_0\|_1 + \frac{1}{2} \|\PTp H_L\|_*  \\
 & & \quad+ \frac{\lambda}{2} \|\POp H_S\|_1-\frac{\lambda}{4}\|\PO H_S\|_F.
\end{eqnarray*}
We note that
\begin{eqnarray*}
\|\PO H_S\|_F & \leq & \|\PO \PG H_S\|_F + \|\PO \PGp H_S\|_F \\
& \leq & \|\PO \PG H_L\|_F + \frac{1}{2}\|H_S\|_F \\
& \leq & \|\PG H_L\|_F + \frac{1}{2} \|\PO H_S\|_F + \frac{1}{2} \|\POp H_S\|_F \\
& \leq & \|\PTp H_L\|_F + \frac{1}{2} \|\PO H_S\|_F + \frac{1}{2} \|\POp H_S\|_F.
\end{eqnarray*}
In the second step above, we have used the fact that $\PG H_L = \PG H_S$ (since $\Gamma \subseteq Q$), and the final inequality follows from Fact \ref{fact:subproj}. Thus, we have
$$
\|\PO H_S\|_F \leq 2\|\PTp H_L\|_F + \|\POp H_S\|_F \leq 2\|\PTp H_L\|_* + \|\POp H_S\|_1.
$$
Putting it all together, we get
$$
\|L_0 + H_L\|_* + \lambda \|S_0-H_S\|_1 \geq \|L_0\|_* + \lambda \|S_0\|_1 +\frac{1-\lambda}{2}\|\PTp H_L\|_* + \frac{\lambda}{4}\|\POp H_S\|_1.
$$
The desired result follows from the fact that $Q^\perp \cap(T \oplus\Omega) = \{\bz\}$.
\end{proof}

\section{Proof Strategy}
\label{sec:dual_cons}

By Lemma \ref{lem:reldual2}, in order for us to prove either Theorem 1 or 2, it is sufficient to produce a dual certificate $W \in \R^{m \times n}$ satisfying
\begin{equation}
\left \{
\begin{array}{l}
W \in T^\perp, \\
\PQp W = -\PQp(UV^*), \\
\|W\| < 1/2, \\
\|\PO(UV^* - \lambda \sgn(S_0) + W)\|_F \leq \lambda/4, \\
\|\POp(UV^* + W)\|_\infty < \lambda/2.
\end{array}
\right .
\label{eqn:dualcert}
\end{equation}

To prove Theorems 1 and 2 under the above conditions, we try to construct the dual certificate $W$ by following a similar strategy as that in the original PCP \cite{Candes2011-JACM}. However, the extra projection of the observations onto the subspace $Q$ adds significant difficulty to various technical parts of the proof. In this section, we will outline the basic components for constructing such a certificate and then provide detailed proofs for each of the component in next sections.  For simplicity, throughout our discussion below, we set $\Gamma \doteq Q \cap T^\perp$ so that $\Gamma^\perp = Q^\perp \oplus T$.

As the support of the sparse matrix is distributed as $\Omega \sim \ber(\rho)$ for some small $\rho \in (0,1)$. This is, of course, equivalent to assuming that $\Omega^c \sim \ber(1-\rho)$. Suppose that $\Omega_1, \Omega_2,\ldots,\Omega_{j_0}$ are independent support sets such that $\Omega_j \sim \ber(q)$ for all $j$. Then, $\Omega^c$ and $\bigcup_{j = 1}^{j_0} \Omega_j$ have the same probability distribution if $\rho = (1-q)^{j_0}$. We now propose a construction for the dual certificate $W \doteq W^L + W^S + W^Q$ as follows. We use a combination of the golfing scheme proposed in \cite{Gross2009-pp} and the least norm approach.

\begin{enumerate}
\item {\it Construction of $W^L$ using the golfing scheme.} Starting with $Y_0 = \bz$, we iteratively define
\begin{equation}
Y_j = Y_{j-1} + q^{-1} \POj \PGp(UV^* - Y_{j-1}),
\label{eqn:yjs}
\end{equation}
and set
\begin{equation}
W^L = \PG Y_{j_0},
\end{equation}
where $j_0 = \lceil 2\log m \rceil$.

\item {\it Construction of $W^S$ by least norm solution.} We define $W^S$ by the following least norm problem:
\begin{equation}
\begin{array}{ccl}
W^S & = & \quad \arg\min_X \, \|X\|_F \\
\st & & \PO X = \lambda \sgn(S_0) \\
& & \PGp X = \bz.
\end{array}
\end{equation}

\item {\it Construction of $W^Q$ by least squares.} We define $W^Q$ by the following least squares problem:
\begin{equation}
\begin{array}{ccl}
W^Q & = & \quad \arg\min_X \, \|X\|_F \\
\st  & & \PQp X = -\PQp(UV^*)  \\
& & \PP X = \bz,
\end{array}
\end{equation}
where $\Pi = \Omega \oplus T$.
\end{enumerate}

We note that under our assumptions (see Section \ref{sec:assumptions}), both the least squares programs above are feasible with high probability under both the random subspace model and the deterministic subspace model. This is because we will later show that the spectral norms of the linear operators $\PO\PGp$ and $\PQp\PP$ can be bounded below unity with high probability.

Thus, to prove that $W^L + W^S + W^Q$ is a valid dual certificate, we have to establish the following:
\begin{eqnarray}
\label{eq:3w-spectral}
\|W^L + W^S + W^Q\| < 1/2, \\
\|\PO(UV^* + W^L)\|_F \leq \lambda/4, \\
\label{eq:3w-infty}
\|\POp(UV^* + W^L + W^S + W^Q)\|_\infty < \lambda/2.
\end{eqnarray}

\begin{lemma}
Assume that $\Omega \sim \text{Ber}(\rho)$ for some small $\rho \in (0,1)$ and the assumptions \eqref{eqn:incoh_subspace} and \eqref{eqn:nu-coherence} hold true. Then, the matrix $W^L$ obeys, with high probability,
\begin{enumerate}
\item $\|W^L\| < 1/4$,
\item $\|\PO (UV^*+W^L)\|_F < \lambda/4$,
\item $\|\POp (UV^*+W^L)\|_\infty < \lambda/4$.
\end{enumerate}
\label{con:wl}
\end{lemma}

\begin{lemma}
In addition to the assumptions in the previous lemma, assume that the signs of the non-zero entries of $S_0$ are i.i.d. random. Then, the matrix $W^S$ obeys, with high probability,
\begin{enumerate}
\item $\|W^S\| < 1/8$,
\item $\|\POp W^S\|_\infty < \lambda/8$.
\end{enumerate}
\label{con:ws}
\end{lemma}

\begin{lemma}
Assume that $\Omega \sim \text{Ber}(\rho)$ for some small $\rho \in (0,1)$ and the assumptions \eqref{eqn:incoh_subspace} and \eqref{eqn:nu-coherence} hold true. Then, the matrix $W^Q$ obeys, with high probability,
\begin{enumerate}
\item $\|W^Q\| < 1/8$,
\item $\|\POp W^Q\|_\infty < \lambda/8$.
\end{enumerate}
\label{con:wq}
\end{lemma}

\noindent The above lemmas together establish a valid dual certificate that satisfies Eqn. (\ref{eq:3w-spectral}) to Eqn. (\ref{eq:3w-infty}).

\section{Random Reduction: Proof of Theorem \ref{thm:cpcp_random}}
\label{sec:main_random_proof}

In this section, we provide a detailed proof of Lemmas \ref{con:wl}, \ref{con:ws}, and \ref{con:wq} for the case when $Q$ is a random subspace. Before proceeding to the main steps of the proof, we first establish some important properties and relationships among the different quantities involved in the problem.

\subsection{Preliminaries}
\label{sec:prelim_random}


\begin{lemma}
Let $Q^\perp$ be a linear subspace distributed according to the random subspace model described earlier. Then, for any $(i,j) \in [m] \times [n]$, with high probability,
\begin{equation}
\|\PQp \bar{e}_i e_j^*\|_F \leq 4\, \sqrt{\frac{p\log(mnp)}{mn}}.
\end{equation}
\label{lem:qpeiej}
\end{lemma}
\begin{proof}
For any $(i,j) \in [m] \times [n]$, we have
\begin{equation}
\|\PQp \bar{e}_i e_j^*\|_F = \sqrt{\sum_{k = 1}^p |\langle G_k, \bar{e_i} e_j^* \rangle|^2} \leq \sqrt{p} \max_k \, \|G_k\|_\infty.
\label{eqn:pq_eiej}
\end{equation}
We now derive a bound for $\|G_k\|_\infty$. Suppose that $M \in \R^{m \times n}$ is a random matrix whose entries are i.i.d. according to the standard normal distribution. Let us define
$$
H = \frac{1}{\sqrt{mn}} M,
$$
and $G = H/\|H\|_F$. Clearly, $G$ is identical in distribution to $G_1, G_2, \ldots, G_p$. We know that, for any $(i,j) \in [m] \times [n]$,
$$
\Pr[|M_{ij}| > t] \leq \sqrt{\frac{2}{\pi}} \,\frac{e^{-t^2/2}}{t}.
$$
Therefore, using a union bound,we get
$$
\Pr[\|M\|_\infty > t] \leq \sqrt{\frac{2}{\pi}} \,\frac{mn}{t} \, e^{-t^2/2},
$$
or equivalently,
$$
\Pr\left [\|H\|_\infty > \frac{t}{\sqrt{mn}}\right ] \leq \sqrt{\frac{2}{\pi}} \,\frac{mn}{t} \, e^{-t^2/2}.
$$
Now, if we have $p$ random matrices $H_1, H_2, \ldots, H_p$, independent and identical in distribution to $H$, then
$$
\Pr\left [\max_k \, \|H_k\|_\infty > \frac{t}{\sqrt{mn}}\right ] \leq \sqrt{\frac{2}{\pi}} \,\frac{mnp}{t} \, e^{-t^2/2}.
$$
Setting $t = \sqrt{4\log(mnp)}$, we get
$$
\Pr\left [\max_k \, \|H_k\|_\infty > \sqrt{\frac{4\log(mnp)}{mn}}\right ] \leq \sqrt{\frac{1}{2\pi}} \,\frac{1}{mnp\sqrt{\log(mnp)}}.
$$
Thus, with high probability, we have that
$$
\max_k\, \|H_k\|_\infty \leq \sqrt{\frac{4\log(mnp)}{mn}}.
$$
It can be shown that $\|H_k\|_F \geq 1/2$ with high probability. Thus, we have that
$$
\max_k\, \|G_k\|_\infty \leq \sqrt{\frac{16\log(mnp)}{mn}},
$$
with high probability. The desired result follows from Eqn. \eqref{eqn:pq_eiej}.
\end{proof}

\begin{lemma}
Assume that $p < mn/4$. Let $Q^\perp$ be a linear subspace distributed according to the random subspace model. Then, with high probability, we have
\begin{equation}
\|\PQp \PT\| \leq 8 \left( \frac{\sqrt{p}+\sqrt{(m+n)r}}{\sqrt{mn}} \right).
\end{equation}
\label{lem:incoh_QT}
\end{lemma}
\begin{proof}
Firstly, we note that $Q^\perp$ is identical in distribution to a subspace spanned by $p$ independent random matrices, each of whose entries are i.i.d. according to a Gaussian distribution with mean zero and variance $1/mn$. Let $\H : \R^p \rightarrow \R^{m \times n}$ be a linear operator defined as follows:
$$
\H(\x) = \sum_{k=1}^p x_k \, H_k,
$$
where the $H_k$'s are independent random matrices each of whose entries are i.i.d. according to a Gaussian distribution with mean zero and variance $1/mn$. Then, we have that $\PQp$ has the same distribution as the operator $\H (\H^* \H)^{-1} \H^*$. Therefore, we have
\begin{eqnarray*}
&& \Pr\left[ \|\PQp \PT\| > 8 \left( \sqrt{\frac{p}{mn}} + \sqrt{\frac{(m+n)r}{mn}}\right) \right] \\
& = & \Pr\left[ \| \H (\H^* \H)^{-1} \H^* \PT\| > 8 \left( \sqrt{\frac{p}{mn}} + \sqrt{\frac{(m+n)r}{mn}}\right) \right] \\
& \leq & \Pr\left[ \| \H (\H^* \H)^{-1}\| \|\H^* \PT\| > 8 \left( \sqrt{\frac{p}{mn}} + \sqrt{\frac{(m+n)r}{mn}}\right) \right] \\
& \leq & \Pr\left[ \| \H (\H^* \H)^{-1}\| > 4\right] + \Pr \left[ \|\H^* \PT\| > 2\left(\sqrt{\frac{p}{mn}} + \sqrt{\frac{(m+n)r}{mn}} \right) \right].
\end{eqnarray*}

Suppose that $R \in \R^{mn \times p}$ is a random matrix whose entries are i.i.d. according to a Gaussian distribution with mean zero and variance $1/mn$. It is easy to see that if we vectorize all the matrices, then $R$ is the matrix analogue of the operator $\H$. Therefore, $\|\H (\H^* \H)^{-1}\|$ has the same distribution as $(\sigma_\mathrm{min}(R))^{-1}$. Let $R' = \sqrt{mn}\,R$. Clearly, the entries of $R'$ are i.i.d according to the standard normal distribution. Using the concentration results for 1-Lipschitz functions (see Proposition 2.18 in \cite{Ledoux}) and the distribution of singular values of random Gaussian matrices \cite{Vershynin2010}, it is possible to show that
$$
\Pr\left[ \sigma_\mathrm{min}(R') \leq \sqrt{mn}-\sqrt{p} - t \right] \leq e^{-t^2/2},
$$
for any $t \geq 0$. Consequently, we have that
$$
\Pr\left[ \sigma_\mathrm{min}(R) \leq 1-\sqrt{\frac{p}{mn}} - t \right] \leq e^{-mnt^2/2}.
$$
Setting $t = 1/4$ and by our assumption that $p < mn/4$, we get
$$
\Pr\left[ \sigma_\mathrm{min}(R) \leq \frac{1}{4} \right] = \Pr\left[ \|\H (\H^* \H)^{-1}\| \geq 4 \right] \leq e^{-mn/32}.
$$
We now note that $\|\H^* \PT\| = \|\PT \H\|$ is identical in distribution to $\|M\|$, where $M \in \R^{(m+n)r\times p}$ is a random matrix whose entries are i.i.d. $\N(0,1/mn)$. This is because the isotropic Gaussian distribution is rotation-invariant. Hence, without any loss of generality we can assume that the operator $\PT$ preserves only the first $\dim(T) = (m+n)r$ components of the basis elements $H_1,\ldots,H_p$. Once again, invoking Proposition 2.18 in \cite{Ledoux}, we can show that
$$
\Pr\left[ \|M\| \geq \frac{\sqrt{p}+\sqrt{(m+n)r}}{\sqrt{mn}} + t \right] \leq e^{-mnt^2/2}.
$$
Setting $t = \max\left\{\sqrt{p/mn}\, , \sqrt{(m+n)r/mn}\right\}$, it follows that
\begin{eqnarray*}
&& \Pr\left[ \|M\| \geq 2\left(\frac{\sqrt{p}+\sqrt{(m+n)r}}{\sqrt{mn}}\right) \right] \\
&= & \Pr\left[ \|\H^*\PT\| \geq 2\left(\frac{\sqrt{p}+\sqrt{(m+n)r}}{\sqrt{mn}} \right )\right]  \\
&\leq &\min\left\{e^{-p/2},e^{-(m+n)r/2} \right\}.
\end{eqnarray*}
Putting it all together, we get
\begin{eqnarray*}
&&\Pr\left[ \|\PQp \PT\| > 8 \left( \sqrt{\frac{p}{mn}} + \sqrt{\frac{(m+n)r}{mn}}\right) \right]\\
& \leq &e^{-mn/32} +\min\left\{e^{-p/2},e^{-(m+n)r/2} \right\}.
\end{eqnarray*}
Thus, we have that
$$
\|\PQp \PT\| \leq 8 \left( \frac{\sqrt{p}+\sqrt{(m+n)r}}{\sqrt{mn}} \right)
$$
with high probability.
\end{proof}

\begin{lemma}
Let $Q^\perp$ be a linear subspace distributed according to the random subspace model and $\Omega \sim \ber(\rho)$. Then, with high probability, we have
\begin{equation}
\|\PQp \PO\| \leq 8\left( \sqrt{\frac{p}{mn}} + \sqrt{\frac{5\rho}{4}} \right).
\end{equation}
\end{lemma}
\begin{proof}
Proceeding along the same lines of the proof of the previous lemma and conditioned on $\Omega$, we get
\begin{eqnarray*}
&& \Pr\left[ \|\PQp \PO\| > 8 \left( \sqrt{\frac{p}{mn}} + \sqrt{\frac{5\rho}{4}} \right) \:\: \Big | \:\: |\Omega| \leq \frac{5}{4} \rho mn \right] \\
&\leq & e^{-mn/32} + \min\left\{ e^{-p/2}, e^{-5mn\rho/8} \right \}.
\end{eqnarray*}
Using Bernstein's inequality, it is possible to show that
$$
\Pr\left[ |\Omega| > mn\rho (1+\delta) \right] \leq 2\exp\left( - \frac{mn\rho \delta^2}{1-\rho+\frac{2\delta}{3}}\right) \leq 2\exp\left( -\frac{3}{5} mn\rho\delta^2\right),
$$
for any $\delta \in (0,1)$. We set $\delta = 1/4$.
Thus, we have
\begin{eqnarray*}
& & \Pr\left[  \|\PQp \PO\| > 8 \left( \sqrt{\frac{p}{mn}} + \sqrt{\frac{5\rho}{4}} \right) \right] \\
& \leq & \Pr\left[ \|\PQp \PO\| > 8 \left( \sqrt{\frac{p}{mn}} +\sqrt{\frac{5\rho}{4}} \right) \:\: \Big | \:\: |\Omega| \leq \frac{5}{4}\rho mn \right] + \Pr\left[ |\Omega| > \frac{5}{4}\rho mn\right] \\
& \leq & e^{-mn/32} +  \min\left\{ e^{-p/2}, e^{-5mn\rho/8} \right \} + 2\,e^{-3mn\rho/80}.
\end{eqnarray*}
Thus, we have that
$$
\|\PQp \PO\| \leq 8\left( \sqrt{\frac{p}{mn}} + \sqrt{\frac{5\rho}{4}} \right)
$$
with high probability.
\end{proof}

\begin{lemma}
Let $\Omega \sim \ber(\rho)$. Then, with high probability,
\begin{equation}
\|\PO\PT\|^2 \leq \rho +\epsilon,
\end{equation}
provided that $1-\rho \geq C_0 \epsilon^{-2} \frac{\mu r \log m}{n}$ for some numerical constant $C_0 > 0$.
\label{lem:incoh_OT}
\end{lemma}
\begin{proof}
See Corollary 2.7 in \cite{Candes2011-JACM}.
\end{proof}

We now prove the following two results that would help us establish incoherence relations with subspaces obtained by a direct sum of two incoherent subspaces.

\begin{lemma}
Let $S_1$ and $S_2$ be any two linear subspaces in $\R^{m \times n}$ satisfying $\|\P_{S_1} \P_{S_2}\| \leq \alpha < 1$. We define $S = S_1 \oplus S_2$. Then, for any $X \in \R^{m \times n}$, we have
\begin{equation}\|\P_S X\|_F^2 \leq (1-\alpha)^{-1}(\|\P_{S_1} X\|_F^2 + \|\P_{S_2} X\|_F^2).
\end{equation}
\label{lem:length_sum}
\end{lemma}
\begin{proof}
We denote by $\vec: \R^{m\times n} \rightarrow \R^{mn}$, the operation of converting a matrix to a vector by stacking its columns one below another. Suppose that $d_1$ and $d_2$ are the dimensions of the subspaces $S_1$ and $S_2$, respectively. Then, there exist matrices $B_1 \in \R^{mn \times d_1}$ and $B_2 \in \R^{mn \times d_2}$ whose columns constitute orthonormal bases for $S_1$ and $S_2$, respectively.

Let $M \doteq [B_1\quad B_2]$. Clearly, the columns of $M$ constitute a basis for the subspace $S$ in $\R^{mn}$. Hence, for any $X \in \R^{m \times n}$, its projection onto $S$ can be expressed as follows:
$$
\vec(\P_S X) = M (M^* M)^{-1} M^* \vec(X).
$$
We note that $\|B_1^* \vec(X)\|_2 = \|\P_{S_1} X\|_F$ and $\|B_2^* \vec(X)\|_2 = \|\P_{S_2} X\|_F$. Therefore,we have
\begin{eqnarray*}
\|\P_S X\|_F^2 & = & \|\vec(\P_S X)\|_2^2 \\
& = & \|M (M^* M)^{-1} M^* \vec(X)\|_2^2 \\
& \leq & \|M  (M^* M)^{-1}\|^2 \cdot \|M^*\vec(X)\|_2^2 \\
& = & \|M  (M^* M)^{-1}\|^2 \cdot (\|\P_{S_1} X\|_F^2 + \|\P_{S_2} X\|_F^2)
\end{eqnarray*}
Let $M^\dag \doteq (M^*M)^{-1}M^*$ denote the Moore-Penrose pseudoinverse of $M$. It is evident that $\|M^\dag\| = \|M  (M^* M)^{-1}\|$. But we know that $\|M^\dag\| = (\sigma_{\mathrm{min}}(M))^{-1}$, where $\sigma_{\mathrm{min}}(M)$ is the smallest non-zero singular value of $M$. Using the fact that $B_1$ and $B_2$ have orthonormal columns, we can show that $(\sigma_{\mathrm{min}}(M))^2  = \lambda_\mathrm{min}(M^*M) \geq 1-\alpha$, where $\lambda_\mathrm{min}(M^*M)$ is the smallest eigenvalue of $M^* M$.\footnote{Since $M$ has full column rank, $M^* M$ is positive definite.} Therefore, we have
$$
\begin{array}{rcl}
\|\P_S X\|_F^2 & \leq & (\sigma_{\mathrm{min}}(M))^{-2}\, (\|\P_{S_1} X\|_F^2 + \|\P_{S_2} X\|_F^2) \\
& \leq &(1-\alpha)^{-1}(\|\P_{S_1} X\|_F^2 + \|\P_{S_2} X\|_F^2).
\end{array}
$$
\end{proof}

\noindent Suppose that $\|\PQp \PT\| < 1/2$.\footnote{From Lemma \ref{lem:incoh_QT} and the assumptions of Theorem \ref{thm:cpcp_random}, this is true with high probability for sufficiently large $m,n$.} Then, it follows that
\begin{equation}
\|\PGp \be_i e_j^*\|_F^2 \leq 4 \left(\frac{8p \log(mnp)}{mn} + \frac{\mu r}{n} \right),
\label{eqn:nu-coherence-random}
\end{equation}
with high probability, for all $(i,j) \in [m] \times [n]$. In other words, with high probability, when $Q^\perp$ is distributed according to the random subspace model, we have that the subspace $\Gamma^\perp$ is $\gamma$-constrained with $\gamma=4 \left(\frac{8p \log(mnp)}{mn} + \frac{\mu r}{n} \right)$. We further note that $\gamma \log m = O(1/\log m)$ under the conditions of Theorem \ref{thm:cpcp_random}. This fact will be used frequently in our proof below.

\begin{lemma}
Let $S_1$, $S_2$ and $S_3$ be any three linear subspaces in $\R^{m \times n}$ satisfying
$\dim(S_1 \oplus S_2 \oplus S_3) = \dim(S_1) + \dim(S_2) + \dim(S_3)$, and $\|\P_{S_1} \P_{S_2}\| \leq \alpha_{1,2} < 1$, $\|\P_{S_2} \P_{S_3}\| \leq \alpha_{2,3} < 1$ and $\|\P_{S_3} \P_{S_1}\| \leq \alpha_{3,1} < 1$. We define $S = S_1 \oplus S_2$. Then, we have
\begin{equation}
\|\P_S \P_{S_3}\| \leq \sqrt{\frac{\alpha_{2,3}^2 + \alpha_{3,1}^2}{1-\alpha_{1,2}}}.
\end{equation}
\label{lem:norm_sum}
\end{lemma}
\begin{proof}
The proof is a simple application of Lemma \ref{lem:length_sum}. We note that, for any $X \in \R^{m \times n}$,
\begin{eqnarray*}
\|\P_S \P_{S_3} X\|_F^2 & \leq & (1-\alpha_{1,2})^{-1} (\|\P_{S_1} \P_{S_3} X\|_F^2 + \|\P_{S_2} \P_{S_3} X\|_F^2 \\
& \leq & (1-\alpha_{1,2})^{-1} (\|\P_{S_1} \P_{S_3}\|^2 + \|\P_{S_2} \P_{S_3}\|^2) \|X\|_F^2\\
& \leq & (1-\alpha_{1,2})^{-1} (\alpha_{3,1}^2 + \alpha_{2,3}^2) \|X\|_F^2.
\end{eqnarray*}
It follows that
$$
\|\P_S \P_{S_3}\| \leq \sqrt{\frac{\alpha_{2,3}^2 + \alpha_{3,1}^2}{1-\alpha_{1,2}}}.
$$
\end{proof}

\begin{lemma}
Let $\Omega \sim \ber(\rho)$ and $\Gamma^\perp$ be $\gamma$-constrained. Then, for any $\epsilon \in (0,1)$, with high probability,
\begin{equation}
\|\PGp -  \rho^{-1}\PGp\PO\PGp\| \leq \epsilon,
\end{equation}
provided that $\rho \geq C\cdot \epsilon^{-2} \gamma \log m$ for some numerical constant $C > 0$.
\label{lem:normbound}
\end{lemma}

\begin{proof}
The proof is very similar to that of Theorem 4.1 in \cite{Candes2008}. We highlight the main steps here. For each $(i,j) \in [m] \times [n]$, we define binary random variables $\delta_{ij}$, each takes value 1 if $(i,j) \in \Omega$, and 0 otherwise. We note that
\begin{eqnarray*}
\PGp \PO \PGp & = & \sum_{ij} \delta_{ij} \, \PGp \be_i e_j^* \otimes \PGp \be_i e_j^*, \\
\E[\PGp \PO \PGp] & = & \rho\,\PGp,
\end{eqnarray*}
where $\otimes$ denotes the outer or tensor product between matrices. Applying a concentration result for operators of the above form, as established in \cite{Candes2007-IP}, we have, with high probability,
\begin{align}
\|\PGp - \rho^{-1} \,\PGp \PO \PGp\| & \leq C' \sqrt{\frac{\log (mn)}{\rho}} \max_{ij} \|\PGp\be_i e_j^*\|_F\\
                                     & \le C' \sqrt{\frac{\gamma \log (mn)}{\rho}} ,
\end{align}
provided that the right hand side is smaller than 1. Here, $C' > 0$ is a numerical constant. The desired result follows by noting that $n \leq m$, and bounding the right hand side by $\epsilon \in (0,1)$.
\end{proof}

\begin{lemma}
Let $Z \in \Gamma^\perp$ be fixed, $\Gamma^\perp$ be $\gamma$-constrained, and $\Omega \sim \ber(\rho)$. Then, with high probability,
\begin{equation}
\|Z - \rho^{-1}\PGp \PO Z\|_\infty \leq \epsilon \|Z\|_\infty,
\end{equation}
provided that $\rho \geq C_0 \cdot \epsilon^{-2} \gamma \log m$ for some numerical constant $C_0 > 64/3$.
\label{lem:infbound}
\end{lemma}

\begin{proof}
Let $\delta_{ij}$ be a sequence of independent Bernoulli random variables such that
$$
\delta_{ij} = \left \{
\begin{array}{ll}
1, & \mathrm{if} \: (i,j) \in \Omega, \\
0, & \mathrm{otherwise.}
\end{array}
\right .
$$
We define $Z' \doteq Z - \rho^{-1}\PGp\PO Z$. Then,
$$
Z' = \sum_{ij} (1-\rho^{-1}\delta_{ij})Z_{ij}\PGp \be_i e_j^*.
$$
For any $(i_0,j_0) \in [m] \times [n]$, we can express $Z'_{i_0 j_0}$ as a sum of independent random variables as shown below:
$$
Z'_{i_0 j_0} = \sum_{ij}R_{ij}, \quad R_{ij} = (1-\rho^{-1}\delta_{ij})Z_{ij}\langle \PGp \bar{e}_i e_j^*, \be_{i_0}e_{j_0}^*\rangle.
$$
It is easy to show that the $R_{ij}$'s are zero-mean random variables with variance given by
$$
\mathrm{Var}(R_{ij}) = (1-\rho)\rho^{-1}|Z_{ij}|^2 \,|\langle \PGp \be_i e_j^*, \be_{i_0}e_{j_0}^*\rangle|^2.
$$
Therefore,
\begin{eqnarray*}
\sum_{ij} \mathrm{Var}(R_{ij}) & = & (1-\rho)\rho^{-1}\sum_{ij}|Z_{ij}|^2 |\langle \PGp \be_i e_j^*, \be_{i_0}e_{j_0}^*\rangle|^2 \\
& \leq & (1-\rho)\rho^{-1}\|Z\|_\infty^2 \sum_{ij}|\langle \be_i e_j^*, \PGp\be_{i_0}e_{j_0}^*\rangle|^2 \\
& = & (1-\rho)\rho^{-1}\|Z\|_\infty^2 \|\PGp \be_{i_0}e_{j_0}^*\|_F^2 \\
& \leq & (1-\rho)\rho^{-1}\gamma \|Z\|_\infty^2 ,
\end{eqnarray*}
where the last inequality holds with high probability.
Furthermore, we have
\begin{eqnarray*}
|R_{ij}| & \leq & \rho^{-1}\|Z\|_\infty |\langle \PGp \be_i e_j^*, \be_{i_0}e_{j_0}^*\rangle| \\
& \leq & \rho^{-1}\|Z\|_\infty \|\PGp \be_i e_j^*\|_F \|\PGp\be_{i_0} e_{j_0}^*\|_F \\
& \leq & \rho^{-1}\gamma \|Z\|_\infty,
\end{eqnarray*}
with high probability.
Thus, using Bernstein's inequality, we obtain
$$
\Pr\left [|Z'_{i_0 j_0}| > \epsilon \|Z\|_\infty\right ] \leq 2 \exp\left( -\frac{\epsilon^2 \rho}{2\gamma \left( \frac{\epsilon}{3}+ 1-\rho\right)}\right).
$$
Choosing $\epsilon < 1$, we can reduce the above expression to
$$
\Pr\left [|Z'_{i_0 j_0}| > \epsilon \|Z\|_\infty\right ] \leq 2 \exp\left( -\frac{3\epsilon^2 \rho}
{8\gamma}\right).
$$
If $\rho \geq C_0 \epsilon^{-2}\gamma \log m$ for some numerical constant $C_0 > 64/3$, then we have
$$
\Pr\left [|Z'_{i_0 j_0}| > \epsilon \|Z\|_\infty \right ] \leq 2 \exp\left( - \frac{3 C_0 \log m}{32}\right).
$$
Applying a union bound, we get
\begin{equation}
\begin{array}{rcl}
\Pr\left [\|Z'\|_\infty > \epsilon \|Z\|_\infty\right ] & \leq & 2mn \exp\left( - \frac{3 C_0 \log m}{32}\right) \\
& \leq & 2m^{\left(2-\frac{3C_0}{32}\right)}.
\end{array}
\end{equation}
Since $C_0 > 64/3$, we obtain the desired result.
\end{proof}

\noindent The following lemma is a restatement of Theorem 6.3 in \cite{Candes2008}.

\begin{lemma}
Let $Z\in \R^{m \times n}$ be fixed, and $\Omega \sim \ber (\rho)$. Then, with high probability,
\begin{equation}
\|Z - \rho^{-1} \PO Z\| \leq C_0' \sqrt{\frac{m \log m}{\rho}} \|Z\|_\infty,
\end{equation}
provided that $\rho \geq C_0' \frac{\log m}{n}$, where $C_0' > 0$ is a numerical constant.
\label{lem:omeganormbound}
\end{lemma}

\subsection{Proof of Lemma \ref{con:wl}}
\label{sec:wl_random_proof}

Before proceeding to the actual proof, we introduce some additional notation. Let $Z_j \doteq UV^* - \PGp Y_j$, where $Y_j$'s are defined in Eqn. \eqref{eqn:yjs}. Evidently, $Z_j \in \Gamma^\perp$ for all $j \geq 0$. The recursive relation between the $Y_j$'s can then be expressed as
\begin{equation}
Z_j = (\PGp - q^{-1} \PGp \POj \PGp) Z_{j-1}, \quad Z_0 = UV^*.
\end{equation}
Let us assume that $\epsilon \in (0,e^{-1})$. From Lemma \ref{lem:infbound}, we have that
$$
\|Z_j\|_\infty \leq \epsilon \|Z_{j-1}\|_\infty,
$$
with high probability, provided that
\begin{equation}
q \geq C_0 \epsilon^{-2}\gamma \log m,
\label{eqn:qcons}
\end{equation}
where $C_0 > 64/3$ is a numerical constant. Since $Z_0 = UV^*$, with high probability, we have
$$
\begin{array}{rcl}
\|Z_j\|_\infty & \leq & \epsilon^j \|UV^*\|_\infty \\
& \leq & \epsilon^j \sqrt{\frac{\mu r}{mn}}.
\end{array}
$$
The second inequality above follows from our assumptions about the matrices $U$ and $V$. Furthermore, when Eqn. \eqref{eqn:qcons} holds, we also have, with high probability,
\begin{equation}
\|Z_j\|_F \leq \epsilon \|Z_{j-1}\|_F
\end{equation}
using Lemma \ref{lem:normbound}. Once again, since $Z_0 = UV^*$, we deduce that
\begin{equation}
\begin{array}{rcl}
\|Z_j\|_F & \leq & \epsilon^j \|UV^*\|_F \\
& = & \epsilon^j \sqrt{r}
\end{array}
\end{equation}
with high probability.

\subsubsection{Bounding $\|W^L\|$}
\label{sub:|WL|}

We first introduce a few notions before deriving a bound on $\|W^L\|$. We let $R$ denote the linear subspace obtained by projecting all the points in $Q^\perp$ onto $T^\perp$. By a slight abuse of notation, we denote this by
\begin{align}
\label{eq:R}
R = \PTp Q^\perp.
\end{align}
We note that if $Q^\perp$ is a random $p$-dimensional subspace in $\R^{m \times n}$, then with probability one, $R$ is a $p$-dimensional subspace of $T^\perp$. It is easy to verify that for any $X \in \R^{m \times n}$, we have
\begin{eqnarray*}
\PGp X & = & \PT X + \PR X.
\end{eqnarray*}

We note that
\begin{equation}
Y_{j_0} = \sum_{j=1}^{j_0} q^{-1} \POj Z_{j-1}.
\end{equation}
Thus, we have
\begin{eqnarray*}
\|W^L\| & = & \|\PG Y_{j_0}\| \\
& \leq & \sum_{j=1}^{j_0} \left\|q^{-1} \PG \POj \PGp Z_{j-1}\right\| \\
& = & \sum_{j=1}^{j_0} \left\|\PG(q^{-1} \POj -\id)\PGp Z_{j-1}\right\| \\
& \leq & \sum_{j=1}^{j_0} \left\|\PGp(q^{-1} \POj -\id)\PGp Z_{j-1}\right\| + \sum_{j=1}^{j_0} \left\|(q^{-1} \POj -\id)\PGp Z_{j-1}\right\|.
\end{eqnarray*}
The second term in the above inequality can be bounded with high probability using Lemma \ref{lem:omeganormbound} as follows:
\begin{eqnarray*}
\sum_{j=1}^{j_0} \left\|(q^{-1} \POj -\id)\PGp Z_{j-1}\right\| & \leq & C_0' \sqrt{\frac{m \log m}{q}} \sum_{j=1}^{j_0}\|Z_{j-1}\|_\infty \\
& \leq & C_0' \sqrt{\frac{m \log m}{q}} \sum_{j=1}^{j_0} \epsilon^{j-1} \sqrt{\frac{\mu r}{mn}} \\
& \leq & C_0' \sqrt{\frac{\mu r \log m}{q\,n}} (1-\epsilon)^{-1},
\end{eqnarray*}
provided that
$$
q \geq \max\left \{C_0' \frac{ \log m}{n},  C_0 \epsilon^{-2} \gamma \log m\right \}.
$$
On the other hand, each term in the summation in the first term can be split as
\begin{eqnarray*}
&&\left\|\PGp(q^{-1} \POj -\id)\PGp Z_{j-1}\right\|\\
 & \leq & \left\|\PT(q^{-1} \POj -\id)\PGp Z_{j-1}\right\| + \left\|\PR(q^{-1} \POj -\id)\PGp Z_{j-1}\right\| \\
& \leq & 2 \left\|(q^{-1} \POj -\id) \PGp Z_{j-1}\right\| + \left\|\PR(q^{-1} \POj -\id)Z_{j-1}\right\|.
\end{eqnarray*}
We have already seen how the first term in the above inequality can be bounded with high probability. Hence, we now focus on the second term. We first state the matrix Bernstein inequality (see Theorem 1.4 in \cite{Tropp2011-FCM}) that will enable us to derive a bound on the second term.


\begin{thm}[Matrix Bernstein Inequality]
Let $M_1, \ldots, M_k \in \R^{d_1 \times d_2}$ be $k$ independent random matrices satisfying
\begin{equation}
\E[M_i] = \bz, \quad \|M_i\| \leq S \:\: \mathrm{almost}\:\: \mathrm{surely}, \quad i = 1,\ldots,k.
\end{equation}
We set
\begin{equation}
\sigma^2 = \max \left \{ \left\| \sum_{i=1}^k \E[M_i^* M_i]\right \|, \left\| \sum_{i=1}^k \E[M_i M_i^*]\right \|\right \}.
\end{equation}
Then, for any $t > 0$, we have
\begin{equation}
\Pr\left[\left\|\sum_{i=1}^k M_i \right \| > t\right] \leq (d_1 + d_2) \exp \left( -\frac{t^2}{2\sigma^2 + 3St}\right).
\end{equation}
\label{thm:operatorbound}
\end{thm}

Using Theorem \ref{thm:operatorbound}, we will now show that, with high probability,
$$
\left\|\PR(q^{-1} \POj -\id)Z_{j-1}\right\| \le  \tilde{C} p \sqrt{ m} \log m \|Z_{j-1}\|_\infty.
$$
The proof is as follows.

For every $(i,l) \in [m] \times [n]$, let us define $M_{il} \doteq H_{il} (Z_{j-1})_{il}\PR \bar{e}_ie_l^*$, where the $H_{il}$'s are independent random variables distributed as follows:
$$
H_{il} = \left \{
\begin{array}{ll}
1, & \mathrm{w.p.} \:\:1-q \\
1-q^{-1}, & \mathrm{w.p.} \:\: q
\end{array}
\right .
$$
We note that $\sum_{i,l} M_{il}$ has the same distribution as $\PR(\id - q^{-1}\POj)Z_{j-1}$. Since the $H_{il}$'s are independent zero-mean random variables that are independent of $Z_{j-1}$, we have that, for any $(i,l) \in [m] \times [n]$,
$$
\E[M_{il} \, | \, Z_{j-1}] = \bz.
$$
We record two useful bounds. We have that
\begin{equation}
1-\rho = \Pr\left[ \cup_j \left\{(i,l) \in \Omega_j\right\} \right] \;\le\; j_0 q.
\end{equation}
So $q \ge (1-\rho)/j_0$. Since $|H_{il}| \le q^{-1}$ almost surely, and $j_0 \ge C/\log m$, we have
\begin{equation}
|H_{il}| \leq O(\log m) \quad \mathrm{almost} \;\mathrm{surely}.
\end{equation}
We also have
\begin{equation}
\|\PR \bar{e}_i e_l^* \| \leq \|\PR \bar{e}_i e_l^* \|_F \leq 1,
\end{equation}
for any $(i,l) \in [m] \times [n]$. It follows that $\|M_{il}\| \leq O(\log m) \|Z_{j-1}\|_\infty$ almost surely.

Now we bound the variance term. It can be shown that $\E[H_{il}^2] = O(\log m)$. Let $B_1,\ldots,B_p$ be such an orthonormal basis for $R$. Then, we have
\begin{align}
\left \| \sum_{il} \E[M_{il} M_{il}^*]\right \| & =  \left \| \sum_{il} \E[H_{il}^2] \PR[\bar{e}_ie_l^*] (\PR[\bar{e_i} e_l^*])^* (Z_{j-1})_{il}^2 \right \| \\
& \leq  O(\log m) \|Z_{j-1}\|_\infty^2 \left \| \sum_{il} \PR[\bar{e}_i e_l^*] ( \PR[\bar{e}_i e_l^*] )^* \right \| \\
& =  O(\log m) \|Z_{j-1}\|_\infty^2 \left \| \sum_{il} \left (\sum_{s=1}^p B_s \langle B_s, \bar{e}_i e_l^*\rangle \right) \left (\sum_{t=1}^p B_t \langle B_t, \bar{e}_i e_l^*\rangle \right)^*\right \| \\
& =  O(\log m) \|Z_{j-1}\|_\infty^2 \left \| \sum_{t} B_t B_t^*\right \| \\
\label{eq:Gincoherent}
& \leq  O(\log m) p \|Z_{j-1}\|_\infty^2
\end{align}
A similar bound holds for the other variance term $\E[M_{il}^* M_{il}]$.

Now, using the matrix Bernstein inequality, we have
\begin{eqnarray*}
&&\Pr\left( \|\PR(q^{-1}\POj - \id)Z_{j-1}\| > t | Z_{j-1}, Q \right) \\
&\leq & (m+n) \exp\left ( - \frac{t^2}{C_1 p \log m \|Z_{j-1}\|_\infty^2 + C_2 \log m \|Z_{j-1}\|_\infty t}\right).
\end{eqnarray*}
Therefore, removing the conditioning, we have that, with high probability,
$$
\|\PR(q^{-1}\POj - \id)Z_{j-1}\| \leq \tilde{C} \sqrt{ m} \log m \|Z_{j-1}\|_\infty,
$$
for any $j$, for some numerical constant $\tilde{C} > 0$.

Thuswe have that, with high probability,
\begin{eqnarray*}
\sum_{j=1}^{j_0} \|\PR(q^{-1}\POj - \id)Z_{j-1}\| & \leq & \sum_{j=1}^{j_0} \tilde{C} \sqrt{ m} \log m \|Z_{j-1}\|_\infty \\
& \leq &  \tilde{C} \sqrt{ m} \log m \sqrt{\frac{\mu r}{mn}} (1-\epsilon)^{-1}\\
& \leq &  \tilde{C} \log m \sqrt{\frac{\mu r}{n}} (1-\epsilon)^{-1}.
\end{eqnarray*}

Under the assumptions of Theorem \ref{thm:cpcp_random}, the bound on the right hand side can be made arbitrarily small. This gives us the desired bound.


\subsubsection{Bounding $\|\PO(UV^* + W^L)\|_F$}
\label{sec:POUVWL}
We now prove the second part of Lemma \ref{con:wl}. First, we note that $\PO Y_{j_0} = \bz$ by construction. Therefore,
\begin{equation}
\PO(UV^* + \PG Y_{j_0}) = \PO(UV^* - \PGp Y_{j_0}) = \PO Z_{j_0}.
\end{equation}
Consequently, we have
$$
\begin{array}{rcl}
\|\PO(UV^* + \PG Y_{j_0})\|_F & = & \|\PO Z_{j_0}\|_F \\
& \leq & \|Z_{j_0}\|_F \\
& \leq & \epsilon^{j_0} \sqrt{r} \\
& \leq & \frac{\sqrt{r}}{m^2}.
\end{array}
$$
The last step follows from the fact that $\epsilon < e^{-1}$ and $j_0 \geq 2\log m$.

\subsubsection{Bounding $\|\POp(UV^* + W^L)\|_\infty$}
\label{sec:POpUVWL}
We now prove the final part of Lemma \ref{con:wl}. We note that
\begin{equation}
UV^* + W^L = UV^* + \PG Y_{j_0} = Y_{j_0} + Z_{j_0}.
\end{equation}
Since we have already proved that $\|Z_{j_0}\|_F < \lambda / 8$, it is sufficient to show that $\|Y_{j_0}\|_\infty < \lambda / 8$. We have
\begin{eqnarray*}
\|Y_{j_0}\|_\infty & = & \| \sum_{j=1}^{j_0} q^{-1} \POj Z_{j-1} \|_\infty \\
& \leq & q^{-1}\sum_{j=1}^{j_0} \|\POj Z_{j-1} \|_\infty \\
& \leq & q^{-1}\sum_{j=1}^{j_0} \|Z_{j-1} \|_\infty  \\
& \leq & q^{-1}\sum_{j=1}^{j_0} \epsilon^{j-1} \sqrt{\frac{\mu r}{mn}}
\end{eqnarray*}
Notice that $q\ge (1-\rho)/j_0\ge 4/\log m$ for $\rho<1/2$,
\begin{eqnarray*}
\|Y_{j_0}\|_\infty & \leq & \frac{\log m}{4(1-\epsilon)}\sqrt{\frac{\mu r}{mn}}\\
& \leq & \lambda/8,
\end{eqnarray*}
for sufficiently small $\epsilon$ and for some numerical constant $C_r$.


\subsection{Proof of Lemma \ref{con:ws}}
\label{sec:ws_random_proof}
We recall the notation that $\Gamma^\perp = Q^\perp \oplus T$. By Lemma \ref{lem:norm_sum}, we have that
$$
\|\PGp \PO\|^2 \leq \frac{\|\PO\PQp\|^2 + \|\PO \PT\|^2}{1 - \|\PQp \PT\|}.
$$
Let us assume that $m, n$ are sufficiently large so that the following conditions hold true:
\begin{eqnarray}
\frac{p}{mn} & < & \frac{5\rho}{4}, \\
8 \left( \sqrt{\frac{p}{mn}} + \sqrt{\frac{(m+n)r}{mn}}\right) & < & \frac{1}{2}, \\
\rho^2(1-\rho) & \geq & C_0 \cdot \frac{\mu r \log m}{n},
\end{eqnarray}
where $C_0 > 0$ is the numerical constant from Lemma \ref{lem:incoh_OT}. We also assume that $\rho < 1/5$. We note that it is possible to satisfy all of the above inequalities under the assumptions on $p$ and $r$ given in Theorem \ref{thm:cpcp_random}, and because $\rho$ is a fixed constant in the interval $(0,1)$. Using Lemma \ref{lem:norm_sum}, it is easy to verify that under these assumptions, with high probability, we have that
\begin{equation}
\|\PGp \PO\| \leq \eta \sqrt{\rho},
\end{equation}
where $\eta > 0$ is a numerical constant.

The basic steps of the proof closely follow that of Lemma 2.9 in \cite{Candes2011-JACM}. We recognize that using the convergent Neumann series, $W^S$ can be expressed as follows:
\begin{equation}
W^S = \lambda (\id - \PGp)\PO\sum_{k \geq 0} (\PO \PGp \PO)^k [\sgn(S_0)].
\end{equation}
As mentioned in Section \ref{sec:assumptions}, we assume that the signs of the non-zero entries of $S_0$ are independent, symmetric $\pm 1$ random variables.


\subsubsection{Bounding $\|W^S\|$}
\label{sec:WS}

It is easy to show that
\begin{eqnarray*}
W^S & = &  \lambda \sgn(S_0) - \lambda \POp \PGp \PO \sum_{k \geq 0} (\PO \PGp \PO)^k [\sgn(S_0)] \\
& : = & W^S_1 - W^S_2
\end{eqnarray*}
We now show that each of these components have spectral norm smaller than $1/16$ with high probability. This gives us the desired bound on $\|W^S\|$.

For the first term, we can use standard arguments about the norms of random matrices with i.i.d. entries (see \cite{Vershynin2011}) to show that, with high probability,
$$
\|\sgn(S_0)\| \leq 4 \sqrt{n\rho}.
$$
Since $\lambda = m^{-1/2}$, we have that $\|W^S_1\| \leq 4\sqrt{\rho}$ with high probability. Thus, for sufficiently small $\rho$, we have that $\|W^S_1\| < 1/16$.

We use a discretization argument to bound $\|W^S_2\|$. Let $N_m$ and $N_n$ be $1/2$-nets for the unit spheres in $\R^m$ and $\R^n$, respectively. It can be shown that the sizes of $N_m$ and $N_n$ are at most $6^m$ and $6^n$, respectively (see Theorem 4.16 in \cite{Ledoux}). Then, we have that
\begin{eqnarray*}
\|W^S_2\| & \leq & 4 \, \max_{\x \in N_m, \y \in N_n} \, \x^* W^S_2 \y \\
& = & 4 \, \max_{\x \in N_m, \y \in N_n} \, \left \langle \x \y^*, W^S_2 \right \rangle \\
& = & 4\, \max_{\x \in N_m, \y \in N_n} \, \left \langle \x \y^*, \lambda \POp \PGp\PO\sum_{k \geq 0} (\PO \PGp \PO)^k [\sgn(S_0)] \right \rangle \\
& = & 4\lambda \,  \max_{\x \in N_m, \y \in N_n} \, \left \langle \sum_{k \geq 0} (\PO \PGp \PO)^k \PO \PGp \POp [\x \y^*], \sgn(S_0) \right \rangle \\
& = & 4\lambda \,  \max_{\x \in N_m, \y \in N_n} \, \left \langle H(\x,\y), \sgn(S_0) \right \rangle.
\end{eqnarray*}
For any $(\x,\y) \in N_m \times N_n$, we bound $\|H(\x,\y)\|_F$ as follows:
\begin{eqnarray*}
\|H(\x,\y)\|_F & = & \left \| \sum_{k \geq 0} (\PO \PGp \PO)^k \PO \PGp \POp [\x \y^*] \right \|_F \\
& \leq & \sum_{k \geq 0} \left \| (\PO \PGp \PO)^k \PO \PGp \POp [\x \y^*]  \right \|_F \\
& \leq & \left (\sum_{k\geq 0} \|(\PO \PGp \PO)\|^k\right) \|\PO\PGp\| \|\POp[\x\y^*]\|_F \\
& \leq & \frac{\|\PO \PGp\|}{1-\|\PO\PGp\|^2}.
\end{eqnarray*}
Conditioned on $Q$ and $\Omega$, we use Hoeffding's inequality to get
$$
\Pr\left[ \left |\left\langle H(\x,\y), \sgn(S_0) \right \rangle \right |  > t | \Omega, Q\right] < 2\exp\left( -\frac{2t^2}{\|H(\x,\y)\|_F^2}\right).
$$
Subsequently, using a union bound over $N_m \times N_n$, we obtain
\begin{align}
& \Pr\left[ \max_{\x \in N_m, \y \in N_n} \left |\left\langle H(\x,\y), \sgn(S_0) \right \rangle \right |  > t | \Omega, Q\right]  \\
& <  2\cdot 6^{m+n}\cdot\exp\left( -\frac{2t^2}{\max_{\x \in N_m, \y \in N_n} \|H(\x,\y)\|_F^2}\right) \\
& \leq  2\cdot 6^{m+n}\cdot\exp\left( -\frac{2t^2 (1-\|\PO \PGp\|^2)^2}{\|\PO\PGp\|^2}\right ).
\label{eq:epsilonnet}
\end{align}
Let $E_1$ be the event $\{\|\PO \PGp\| \leq \eta \sqrt{\rho}\}$. We know that this event occurs with high probability. Thus, removing the conditioning on $\Omega$ and $Q$, we have
\begin{eqnarray*}
\Pr\left[ \max_{\x \in N_m, \y \in N_n} \left |\left\langle H(\x,\y), \sgn(S_0) \right \rangle \right |  > t \right] &< & 2\cdot 6^{m+n}\cdot\exp\left( -\frac{2t^2 (1-\eta^2\rho)^2}{\eta^2 \rho}\right ) \\
&& \quad+ \Pr[E_1^c].
\end{eqnarray*}
Therefore,
\begin{eqnarray*}
\Pr\left[ 4\lambda \max_{\x \in N_m, \y \in N_n} \left |\left\langle H(\x,\y), \sgn(S_0) \right \rangle \right |  > t \right] &< & 2\cdot 6^{m+n}\cdot\exp\left( -\frac{t^2 (1-\eta^2\rho)^2}{8\lambda^2\eta^2 \rho}\right ) \\
&& \quad+ \Pr[E_1^c].
\end{eqnarray*}
Setting $t = \frac{s\eta\sqrt{16\rho}}{1-\eta^2\rho}$ and substituting $\lambda = 1/\sqrt{m}$, we get
$$
\Pr\left[ \|W^S_2\| > \frac{s\eta\sqrt{16\rho}}{1-\eta^2\rho}\right ] < 2\cdot\exp\left ( 2m (\log 6 - s^2) \right) + \Pr[E_1^c].
$$
Let us choose any $s > \sqrt{\log 6}$. Then, for sufficiently small $\rho$, we have that $\|W^S_2\|< 1/16$ with high probability.

\subsubsection{Bounding $\|\POp W^S\|_\infty$}
\label{sec:ws_infty}

Once again, using the convergent Neumann series expansion for $W^S$, we have
\begin{eqnarray*}
\|\POp W^S\|_\infty & = & \max_{(i,j) \in \Omega^c} \left | \left \langle \bar{e}_i e_j^*,  \lambda (\id - \PGp)\PO\sum_{k \geq 0} (\PO \PGp \PO)^k [\sgn(S_0)] \right \rangle \right | \\
& = & \lambda \max_{(i,j) \in \Omega^c} \left | \left \langle \sum_{k \geq 0} (\PO \PGp \PO)^k \PO \PGp [\bar{e}_i e_j^*], \sgn(S_0) \right \rangle \right | \\
& = & \lambda \max_{(i,j) \in \Omega^c} \left | \left \langle X_{i,j}, \sgn(S_0) \right \rangle \right |.
\end{eqnarray*}
Conditioned on $Q$ and $\Omega$, we use Hoeffding's inequality to get
$$
\Pr\left[ \left |\left\langle X_{i,j}, \sgn(S_0) \right \rangle \right |  > t | \Omega, Q\right] < 2\exp\left( -\frac{2t^2}{\|X_{i,j}\|_F^2}\right).
$$
Using a union bound, we obtain
$$
\Pr\left[ \max_{i,j} \left |\left\langle X_{i,j}, \sgn(S_0) \right \rangle \right |  > t | \Omega, Q\right] < 2mn\exp\left( -\frac{2t^2}{\max_{i,j} \|X_{i,j}\|_F^2}\right).
$$
We obtain a bound on $\|X_{i,j}\|_F$ as follows:
\begin{eqnarray*}
\|X_{i,j}\|_F & = & \left \| \sum_{k \geq 0} (\PO \PGp \PO)^k \PO \PGp [\bar{e}_i e_j^*] \right \|_F \\
& \leq & \sum_{k\geq 0} \left \| (\PO \PGp \PO)^k \PO \PGp [\bar{e}_i e_j^*] \right \|_F \\
& \leq & \left( \sum_{k \geq 0} \|\PO\PGp\PO\|^k\right) \|\PO \PGp [\bar{e}_i e_j^*] \|_F \\
& \leq & \frac{\|\PO\PGp\| \|\PGp[\bar{e}_i e_j^*]\|_F}{1 - \|\PO\PGp\|^2}.
\end{eqnarray*}
Thus, we get
\begin{eqnarray*}
&&\Pr\left[ \max_{i,j} \left |\left\langle X_{i,j}, \sgn(S_0) \right \rangle \right |  > t | \Omega, Q\right] \\
&< & 2mn\exp\left( -\frac{2t^2(1 - \|\PO\PGp\|^2)^2}{\|\PO\PGp\|^2 \max_{i,j}\|\PGp[\bar{e}_i e_j^*]\|_F^2}\right).
\end{eqnarray*}
Removing the conditioning on $Q$ and $\Omega$, we get
$$
\Pr\left[ \|\POp W^S\|_\infty > \lambda \sqrt{\frac{s\log(mn)}{2}}\frac{\|\PO\PGp\| \max_{i,j}\|\PGp[\bar{e}_i e_j^*]\|_F}{1 - \|\PO\PGp\|^2 }\right] < 2 (mn)^{1-s}
$$
Consider the two events:
\begin{eqnarray*}
E_1 & := & \left \{\|\PO\PGp\| \leq \eta \sqrt{\rho} \right\}, \\
E_2 & := & \left \{ \max_{i,j} \,\|\PGp \bar{e}_i e_j^*\|_F \leq \sqrt{\gamma} \right \},
\end{eqnarray*}
where we recall that $\gamma=4 \left(\frac{8p \log(mnp)}{mn} + \frac{\mu r}{n} \right)$. We have already shown that $E_1$ and $E_2$ occur with high probability. Substituting for the various bounds and setting $s = 2$, we get
\begin{eqnarray*}
&&\Pr\left[ \|\POp W^S\|_\infty > \lambda \sqrt{\gamma \log(mn)}  \frac{\eta \sqrt{\rho}}{1-\eta^2\rho} \right ] \\ &< & \frac{2}{mn} + \Pr[ (E_1 \cap E_2)^c].
\end{eqnarray*}
Under the conditions of Theorem \ref{thm:cpcp_random}, and for sufficiently large $m, n$ and sufficiently small $\rho$, we get that $ \|\POp W^S\|_\infty < \lambda/8$ with high probability.

\subsection{Proof of Lemma \ref{con:wq}}
\label{sec:wq_random_proof}

\subsubsection{Bounding $||W^Q||$}
\label{sec:wq_spectral}

Using the convergent Neumann series expansion, we can write the analytical expression for $W^Q$ as follows:
\begin{equation}
W^Q = \PPp \sum_{k \geq 0} (\PQp \PP \PQp)^k (\PQp (-UV^*)),
\end{equation}
where we recall that $\Pi = \Omega \oplus T$. It follows that
$$
\|W^Q\|_F \leq \left \|\sum_{k \geq 0} (\PQp \PP \PQp)^k \right \| \|\PQp (UV^*)\|_F.
$$
Considering the first term of the product on the right hand side,
\begin{eqnarray*}
\left \|\sum_{k \geq 0} (\PQp \PP \PQp)^k \right \| & \leq & \sum_{k \geq 0} \left \| (\PQp \PP \PQp)^k \right \| \\
& \leq & \sum_{k \geq 0} \|\PQp \PP\|^{2k}.
\end{eqnarray*}
From Lemma \ref{lem:norm_sum}, we have that, for any $\epsilon > 0$, with high probability,
\begin{eqnarray*}
&&\|\PQp \PP\|^2 \\
&\leq& \frac{64}{1-\sqrt{\rho + \epsilon}} \left(\left( \sqrt{\frac{p}{mn}} + \sqrt{\frac{5\rho}{4}} \right)^2 + \left( \sqrt{\frac{p}{mn}} + \sqrt{\frac{(m+n)r}{mn}} \right)^2 \right).
\end{eqnarray*}
Assume that $\rho < 1/4$, and fix $\epsilon = 3\rho$. For $m,n$ large enough, we can assume that $\max\{ p/mn\, , r(m+n)/mn\} < \rho$. Then, we have that, with high probability,
$$
\|\PQp \PP\|^2 \leq \frac{832\, \rho}{1-2\sqrt{\rho}}.
$$
Therefore, for sufficiently small $\rho$, we have that
\begin{equation}
\|\PQp \PP\|^2 \leq \frac{1}{4},
\end{equation}
with high probability. Consequently,
\begin{equation}
\left \|\sum_{k \geq 0} (\PQp \PP \PQp)^k \right \| \leq \frac{4}{3},
\end{equation}
with high probability.

We bound $\|\PQp (UV^*)\|_F$ as follows. As explained earlier, suppose we vectorize all matrices, then $\PQp$ has the same distribution as $H (H^*H)^{-1} H^*$, where $H \in \R^{mn \times p}$ is a random Gaussian matrix with i.i.d. entries $\sim \N(0,1/mn)$. Therefore, we have
$$
\|\PQp (UV^*)\|_F = \|H (H^*H)^{-1} H^* \vec(UV^*)\|_2 \leq \|H (H^*H)^{-1}\|\, \|H^*\vec(UV^*)\|_2,
$$
where the above equality is in distribution. We have already shown in the proof of Lemma \ref{lem:incoh_QT} that
$$
\Pr\left[\|H (H^*H)^{-1}\| \geq 4 \right] \leq e^{-mn/32}.
$$
We note that $H^*\vec(UV^*)$ is a $p$-dimensional vector whose components are i.i.d. and have the same distribution as $\langle G, UV^* \rangle$, where $G \in \R^{m \times n}$ is a random Gaussian matrix whose entries are i.i.d. $\sim \N(0,1/mn)$. It is easy to see that $\langle G, UV^* \rangle$ is distributed according to $\N(0,r/mn)$, and therefore, we have
$$
\E[\|H^*\vec(UV^*)\|_F] \leq (\E[\|H^*\vec(UV^*)\|_F^2])^{1/2} = \sqrt{\frac{pr}{mn}}.
$$
Since $\|\cdot\|_F$ is a 1-Lipschitz function, we use Proposition 2.18 in \cite{Ledoux} to get
$$
\Pr\left[ \|H^*\vec(UV^*)\|_F  \geq \E(\|H^*\vec(UV^*)\|_F) + t\cdot \sqrt{\frac{r}{mn}} \,\right] \leq e^{-t^2/2}.
$$
Setting $t = \sqrt{6\log m}$, we get
\begin{equation}
\Pr\left( \|H^*\vec(UV^*)\|_F  \geq \sqrt{\frac{pr}{mn}} + \sqrt{\frac{6r\log m}{mn}}\right) \leq \frac{1}{m^3}.
\end{equation}
Putting it all together, we conclude that
\begin{equation}
\|W^Q\| \leq \|W^Q\|_F \leq \frac{16}{3}\left( \sqrt{\frac{pr}{mn}} + \sqrt{\frac{6r\log m}{mn}}\right),
\end{equation}
with high probability. Clearly, for sufficiently large $m$, the right hand side can be made arbitrarily small under the conditions of Theorem \ref{thm:cpcp_random} and hence, we have the desired bound.


\subsubsection{Controlling $\|\POp W^Q\|_\infty$}
\label{sec:wq_infty}

It is easy to show that the analytical expression for $W^Q$ can be written slightly differently as follows:
\begin{equation}
W^Q = \PPp \PQp \sum_{k \geq 0} (\PQp \PP \PQp)^k (\PQp (-UV^*)).
\end{equation}
Consider any $(i,j) \in [m] \times [n]$. Then,
\begin{eqnarray*}
|\langle W^Q, \bar{e}_i e_j^* \rangle| & = & \left |\left \langle \sum_{k \geq 0} (\PQp \PP \PQp)^k (\PQp (-UV^*)),  \PQp \PPp \bar{e}_i e_j^* \right \rangle \right | \\
& \leq & \left \| \sum_{k \geq 0} (\PQp \PP \PQp)^k \PQp(UV^*) \right \|_F \, \|\PQp \PPp \bar{e}_i e_j^*\|_F \\
& \leq & \left \| \sum_{k \geq 0} (\PQp \PP \PQp)^k \right \| \, \|\PQp(UV^*)\|_F \, \|\PQp \PPp \bar{e}_i e_j^*\|_F.
\end{eqnarray*}
We have already derived bounds for the first two terms. For the final term in the product, we use the same technique we employed to bound $\|\PQp(UV^*)\|_F$. Using the fact that $\|\PPp  \bar{e}_i e_j^*\|_F \leq 1$, we can show that
$$
\Pr\left[\|\PQp \PPp \bar{e}_i e_j^*\|_F > \frac{16}{3}\left(\sqrt{\frac{p}{mn}} + \sqrt{\frac{6\log m}{mn}}\right)\right ] \leq \frac{1}{m^3} + e^{-mn/32}.
$$
Using a union bound, we get
$$
\Pr\left[ \max_{i,j}\,\|\PQp \PPp \bar{e}_i e_j^*\|_F > \frac{16}{3}\left(\sqrt{\frac{p}{mn}} + \sqrt{\frac{6\log m}{mn}}\right)\right ] \leq mn(m^{-3} + e^{-mn/32}).
$$
Putting all the bounds together, we have that, with high probability,
\begin{equation}
\|\POp W^Q\|_\infty \leq \frac{256}{9} \frac{\sqrt{r}}{mn} \, \left(\sqrt{p} + \sqrt{6\log m}\right)^2.
\end{equation}
Since $\lambda = m^{-1/2}$, it easy to show that under the assumptions of Theorem \ref{thm:cpcp_random}, the right hand side in the above inequality can be made smaller than $C'\,\lambda$, for any fixed $C' > 0$. Thus, we have the desired bound.


\section{Deterministic Reduction: Proof of Theorem \ref{thm:cpcp_deterministic}}
\label{sec:deterministic_proof}
In this section, we provide the proof for Theorem \ref{thm:cpcp_deterministic} under the deterministic subspace model for $Q^\perp$.
We will adopt the same optimality conditions established in Lemma \ref{lem:reldual2}, and the same proof strategy outlined in Section \ref{sec:dual_cons}, namely the construction of $W=W^L+W^S+W^Q$. To avoid redundancy, wherever possible, we will only highlight the parts that differ from the previous proof in Section \ref{sec:main_random_proof} and refer the interested reader to Section \ref{sec:main_random_proof} for more details. First, we derive the various incoherence relations associated with our fixed subspace $Q^\perp$. Then, we will prove Lemmas \ref{con:wl}, \ref{con:ws} and \ref{con:wq} using these relations.


\subsection{Preliminaries}
In this subsection, we provide several lemmas that will be used later in our proof.
\begin{lemma}
If $X\in \mathbb{R}^{m\times n}$ is a rank-$r$ matrix, then
\begin{align}
\|\PQp X\|_F^2\le \nu \frac{pr}{n} \|X\|_F^2.
\end{align}
\end{lemma}
\begin{proof}
\begin{eqnarray*}
\|\PQp X\|_F^2 & = & \sum_{i=1}^p \left|\braket{G_i,X}\right|^2 \\
& \le & p \left(\max_i \|G_i\|^2\right )\|X\|_*^2\\
& \le & pr \left(\max_i \|G_i\|^2\right ) \|X\|_F^2 \\
&\le & \nu \frac{pr}{n} \|X\|_F^2.
\end{eqnarray*}
\end{proof}

\begin{cor}
For any $\nu$-coherent subspace $Q^\perp$, we have the following:
\label{coro:incoherence_PQp}
\begin{enumerate}
\item $\|\PQp \bar{e}_i e_j^*\|_F^2\le \nu \frac{p}{n}$;
\item $\|\PQp \P_T\|^2\le 2\nu \frac{pr}{n}$;
\item $\|\PQp(UV^*)\|_F^2 \le 2\nu \frac{pr^2}{n}$.
\end{enumerate}
\end{cor}
\begin{proof}
The first two results follow from the fact that the $\bar{e}_i e_j^*$ are rank-$1$ matrices, and $\rank\left(\PT X\right)\le 2r \: \forall X \in \R^{m \times n}$. The last result can be derived from the second one as shown below:
$$
\|\PQp(UV^*)\|_F^2 \le \|\PQp \P_T\|^2 \|UV^*\|_F^2 \le 2\nu \frac{pr^2}{n}.
$$
\end{proof}

\begin{lemma}
Under the assumptions made in Theorem \ref{thm:cpcp_deterministic}, we have that
\begin{align}
\|\PQp \PO\|< 1/2,
\end{align}
with high probability, provided that $\rho< \rho_0$ and $\nu^2 p^3 \log m/n \le C$. Here, $C > 0$  and $\rho_0 \in (0,1)$ are numerical constants.
\label{lemma:PQpPo}
\end{lemma}
\begin{proof}
Please refer to Section \ref{sec:pqp_po} for a detailed proof.
\end{proof}

\subsection{Proof of Lemma \ref{con:wl} (deterministic case)}
We use the same framework from Section \ref{sub:|WL|} to bound the corresponding norms of $W^L$. We note that to bound $\|W^L\|$ in the previous case, the only key property of $Q^\perp$ that was critical to the proof was that $\Gamma^\perp = \Qp\oplus T$ is $O(\mu r/n)$-constrained. More specifically, the latter property is used in Lemma \ref{lem:normbound} and Lemma \ref{lem:infbound}.

In the deterministic case, by assumption, $\Qp$ is $\nu$-coherent, where $\nu$ is a constant. In the following lemma, we will show that $\Gamma^\perp$ is $O(\mu r/n)$-constrained as well under our assumptions. We will show that, the proof of Lemma \ref{con:wl} can be directly adopted for the deterministic case from the that with the random subspace model.
\begin{lemma}
If $\Qp$ is $\nu$-coherent, then
\begin{align}
\|\PGp \bar{e}_i e_j^*\|_F \le 4 \left(\sqrt{\frac{\nu p}{n}}+ \sqrt{\frac{2\mu r}{n}}\right).
\end{align}
In other words, if $\Qp$ is $\nu$-coherent, then $\Gamma^\perp$ is $\gamma$-constrained for $\gamma = 16 \left(\sqrt{\nu p/n}+ \sqrt{2\mu r/n}\right)^2$.
\end{lemma}
\begin{proof}
Let us assume that $\|\PQp \PT\| < 1/2$. This is true for sufficiently large $n$ under the assumptions of Theorem \ref{thm:cpcp_deterministic}. Using the convergent Neumann series expansion, it is possible to show that
$$
\PGp \bar{e}_i e_j^* =\left((\mathcal{I}-\P_\Qp \P_T)^{-1} \P_\Qp \P_{T^\perp} + (\mathcal{I}-\P_T\P_\Qp)^{-1} \P_T \P_{Q}\right)(\bar{e}_i e_j^*),
$$
and therefore,
$$
\|\PGp \bar{e}_i e_j^*\|_F \leq \|\mathcal{I}-\P_\Qp \P_T)^{-1}\| \|\P_\Qp \P_{T^\perp}(\bar{e}_i e_j^*)\|_F + \|\mathcal{I}-\PT \PQp)^{-1}\| \|\PT \PQ (\bar{e}_i e_j^*)\|_F.
$$
From Eqn. \eqref{eq:PTeiej} and Corollary \ref{coro:incoherence_PQp}, we have
\begin{eqnarray*}
\|\P_\Qp \P_{T^\perp}(\bar{e}_i e_j^*)\|_F & \leq & \|\PQp (\bar{e}_i e_j^*)\|_F + \|\P_\Qp \PT(\bar{e}_i e_j^*)\|_F\\
 & \le & \|\P_\Qp (\bar{e}_i e_j^*)\|_F + \|\PT (\bar{e}_i e_j^*)\|_F \\
& \le & \sqrt{\frac{\nu p}{n}}+ \sqrt{\frac{2\mu r}{n}}.
\end{eqnarray*}
Similarly, we have
$$
\|\PT \PQ (\bar{e}_i e_j^*)\|_F \leq \sqrt{\frac{\nu p}{n}}+ \sqrt{\frac{2\mu r}{n}}.
$$
We also have that
$$
\left\|(\mathcal I-\P_\Qp \P_T)^{-1}\right\| = \left\|(\mathcal I-\PT \PQp)^{-1}\right\| = \left\|\sum_{k\ge 0} (\P_\Qp \P_T)^k \right\|                                             < 2.
$$
Hence, we have
\begin{align}
\|\PGp e_i e_j^*\|_F & \le 4 \left(\sqrt{\frac{\nu p}{n}}+ \sqrt{\frac{2\mu r}{n}}\right).
\label{eq:PQp+T}
\end{align}
\end{proof}
It can be easily shown that the results in Section \ref{sec:prelim_random} all hold for the deterministic case as well with the modified value for $\gamma$ derived above. Consequently, the proof of Lemma \ref{con:wl} from Section \ref{sec:wl_random_proof} can be directly adopted for the deterministic case as well.



\subsection{Proof of Lemma \ref{con:ws} (deterministic case)}
We now provide a proof of Lemma \ref{con:ws} under our deterministic subspace model. Since the basic framework of the proof is very similar to that in Section \ref{sec:ws_random_proof}, we will derive only the important steps here and refer the interested reader to Section \ref{sec:ws_random_proof} for more details.

\paragraph{Controlling $\|\POp W^S\|_{\infty}$}
Using the convergent Neumann series, we have
$$
W^S=\lambda \left(\mathcal{I}-\PGp \right)\P_\Omega \sum_{k\ge 0} (\P_\Omega \PGp \P_{\Omega})^k \left[\text{sgn}(S_0)\right].
$$
Therefore, we have
\begin{eqnarray*}
\|\POp W^S\|_{\infty} & = & \lambda \max_{(i,j)\in \Omega^c }\left|\left\langle \bar{e}_ie_j^*, \left(\mathcal{I}-\PGp \right)\P_\Omega \sum_{k\ge 0} (\P_\Omega \PGp \P_{\Omega})^k \left[\text{sgn}(S_0)\right] \right\rangle\right|\\
 & = & \lambda \max_{(i,j)\in \Omega^c }\left|\left\langle \sum_{k\ge 0} (\P_\Omega \PGp \P_{\Omega})^k \PO \PGp (\bar{e}_i e_j^*), \sgn(S_0) \right\rangle\right|\\
 & = & \lambda \max_{(i,j)\in \Omega^c }\left|\left\langle H^{(i,j)}, \sgn(S_0) \right\rangle\right|.
\end{eqnarray*}
We now bound $\|H^{(i,j)}\|_F$ as follows:
\begin{eqnarray*}
\|H^{(i,j)}\|_F & \le & \sum_{k \geq 0} \left \| (\PO \PGp \PO)^k \PO \PGp (\bar{e}_i e_j^*)  \right \|_F \\
& \leq &  \left (\sum_{k\geq 0} \|(\PO \PGp \PO)\|^k\right) \|\PO\PGp\| \| \PGp (\bar{e}_i e_j^*)\|_F \\
& \leq & \frac{\|\PO\PGp \|\|\PGp(\bar{e}_i e_j^*)\|_F}{1-\|\PO\PGp\|^2}.
\end{eqnarray*}
Conditioned on $\Omega$, using Hoeffding's inequality, we have
\begin{eqnarray*}
\Pr\left[\left|\braket{H^{(i,j)}, \sgn(S_0)}\right|>t \ | \ \Omega \right] < 2\exp \left(-\frac{2t^2}{\|H^{(i,j)}\|_F^2}\right).
\end{eqnarray*}
Applying a union bound, we get
\begin{eqnarray*}
\Pr\left[\max_{i,j}\left|\braket{H^{(i,j)},\sgn(S_0)}\right|>t \ | \ \Omega \right] & \le&  2mn \exp\left(-\frac{2t^2}{\max_{i,j} \|H^{(i,j)}\|_F^2}\right)\\
& \le & 2mn \exp\left(-\frac{2t^2\left(1-\|\PO \PGp \|^2\right)^2}{\|\PO \PGp\|^2 \max_{i,j}
\|\PGp(\bar{e}_i e_j^*)\|_F^2 }\right).
\end{eqnarray*}
Removing the conditioning on $\Omega$, we get
$$
\Pr\left[ \|\POp W^S\|_\infty > \lambda \sqrt{\frac{s\log(mn)}{2}}\frac{\|\PO\PGp\| \max_{i,j}\|\PGp(\bar{e}_i e_j^*)\|_F}{1 - \|\PO\PGp\|^2 }\right] < 2 (mn)^{1-s},
$$
where $s > 0$. Consider the event $E = \left \{\|\PO\PGp\| \leq \eta \sqrt{\rho} \right\}$. Just like under the random subspace model, it is not difficult to show that the event $E$ occurs with high probability for some fixed $\eta > 0$. Furthermore, we have already shown that $\Gamma^\perp$ is a $\gamma$-constrained subspace with $\gamma \log m = O(1/\log m)$. Setting $s = 2$, we get
$$
\Pr\left[ \|\POp W^S\|_\infty > \lambda \sqrt{\gamma \log(mn)}  \frac{\eta \sqrt{\rho}}{1-\eta^2\rho} \right ] < \frac{2}{mn} + \Pr[ E^c].
$$
Thus, we have the desired bound.


%
%

\paragraph{Controlling $\|W^S\|$}
The proof is identical to the one in Section \ref{sec:WS}.

\subsection{Proof of Lemma \ref{con:wq} (deterministic case)}
We now prove Lemma \ref{con:wq} under our deterministic subspace model. Once again, the basic structure of the proof is very similar to the one used in Section \ref{sec:wq_random_proof}. So, we only provide the relevant bounds here and refer the interested reader to Section \ref{sec:wq_random_proof} for the detailed steps involved.

\paragraph{Controlling $\|W^Q\|$}
The proof framework is the same as the one in Section \ref{sec:wq_spectral}. We note that the key step is to bound $\|\PQp(UV^*)\|_F$ and $\left \|\sum_{k \geq 0} (\PQp \PP \PQp)^k \right \|$, where we recall that $\Pi = \Omega \oplus T$. From Corollary \ref{coro:incoherence_PQp}, we already know that
$$
\|\PQp(UV^*)\|_F \leq \sqrt{\frac{2 \nu p r^2}{n}}.
$$
For the other quantity, we have that
$$
\left \|\sum_{k \geq 0} (\PQp \PP \PQp)^k \right \| \leq \frac{1}{1 - \|\PQp \PP\|^2}.
$$
By Lemma \ref{lem:norm_sum}, we have
$$
\|\PQp \PP\|^2 \le \frac{\| \PQp \P_\Omega\|^2+\|\PQp\P_T \|^2}{1-\|\P_\Omega \P_T\|}.
$$
From Lemma \ref{lem:incoh_OT}, we know that $\|\PO \PT\| \leq \sqrt{\rho + \epsilon}$ with high probability, provided that
$$
(1-\rho) \geq C_0 \cdot \epsilon^{-2} \frac{\mu r \log m}{n}.
$$
Suppose that the above condition holds with $\epsilon = \rho$, and assume that
$$
\frac{2 \nu p r}{n} < \rho.
$$
We note that both the assumptions above can be true for sufficiently large $m$ and $n$ under the assumptions of Theorem \ref{thm:cpcp_deterministic}. Under these assumptions, along with Lemma \ref{lemma:PQpPo}, we have
$$
\|\PQp \PP\|^2 \le \frac{1/4+\rho}{1-\sqrt{2\rho}},
$$
with high probability. Thus, we have that $\|\PQp \PP\|^2 \le 1/2$ with high probability, provided that $\rho$ is sufficiently small. Putting all these bounds together, we get
$$
\left\|W^Q\right\| \leq 2\sqrt{\frac{2 \nu p r^2}{n}},
$$
with high probability. Under the assumptions of Theorem \ref{thm:cpcp_deterministic}, the right hand side can be made arbitrarily small, and hence, we have the desired result.


\paragraph{Controlling $\|\POp W^Q\|_\infty$} Once again, the proof framework is identical to that used in Section \ref{sec:wq_infty}. The key step here is to bound $\max_{(i,j) \in \Omega^c} \, \|\PQp \PPp \bar{e}_i e_j^*\|_F$. We first use the Neumann series to rewrite $\PP \bar{e}_i e_j^*$ as
$$
\PP \bar{e}_i e_j^* = \left((\mathcal{I}-\P_\Omega \P_T)^{-1} \P_\Omega \P_{T^\perp} + (\mathcal{I}-\P_T\P_\Omega)^{-1} \P_T \P_{\Omega^\perp}\right)(\bar{e}_i e_j^*).
$$
Now, for any $(i,j) \in [m] \times [n]$, we have
\begin{eqnarray*}
\|\PO \PTp \bar{e}_ie_j^*\|_F = \|\PO \PT \bar{e}_i e_j^*\|_F & \leq & \sqrt{\frac{2\mu r}{n}},\\
\|\PT \POp \bar{e}_ie_j^*\|_F = \| \PT \bar{e}_i e_j^*\|_F & \leq & \sqrt{\frac{2\mu r}{n}}.
\end{eqnarray*}
Furthermore, by the assumption we used earlier (to bound $\|W^Q\|$), we have that $\|\PO\PT\| < \sqrt{2\rho}$ with high probability. Therefore, we have
\begin{eqnarray*}
\left \|(\mathcal{I}-\P_\Omega \P_T)^{-1}\right \| & = & \left \|(\mathcal{I}-\P_T \PO)^{-1}\right \| \\
 & = & \left\|\sum_{k\ge 0} (\P_\Omega \P_T)^k \right\| \\
 & \leq & \frac{1}{1 - \sqrt{2\rho}}\\
 & \leq & 2
\end{eqnarray*}
with high probability, provided that $\rho\le 1/8$. Thus, we get
$$
\|\PP \bar{e}_i e_j^*\|_F \leq 4\sqrt{\frac{2\mu r}{n}},
$$
with high probability. Consequently, for any $(i,j) \in \Omega^c$, we have
\begin{eqnarray*}
\|\PQp \PPp \bar{e}_i e_j^*\|_F & \leq  & \|\PQp \bar{e}_i e_j^*\|_F + \|\PQp \PP \bar{e}_i e_j^*\|_F \\
& \leq & \sqrt{\frac{\nu p}{n}} +  4\sqrt{\frac{2\mu r}{n}},
\end{eqnarray*}
with high probability.

Proceeding along the same lines as in Section \ref{sec:wq_infty}, we have that
$$
|\langle W^Q, \bar{e}_i e_j^* \rangle | \leq 2\sqrt{\frac{2 \nu p r^2}{n}} \left( \sqrt{\frac{\nu p}{n}} +  4\sqrt{\frac{2\mu r}{n}}\right ),
$$
with high probability, for any $(i,j) \in \Omega^c$. Therefore, we have that
$$
\|\POp W^Q\|_\infty \leq 2\sqrt{\frac{2 \nu p r^2}{n}} \left( \sqrt{\frac{\nu p}{n}} +  4\sqrt{\frac{2\mu r}{n}}\right ),
$$
with high probability. Since $\lambda = m^{-1/2}$, under the assumptions made in Theorem \ref{thm:cpcp_deterministic}, the right hand side can be made smaller than $\lambda/8$, provided that $n$ is sufficiently large.

\subsection{Proof of Lemma \ref{lemma:PQpPo}}
\label{sec:pqp_po}
Consider the linear operator
$$
\mathcal A = \PQp \PO \PQp - \rho \PQp.
$$
It can be easily shown that
$$
\E \left[\mathcal A\right] = 0.
$$
First, we derive a bound for the spectral norm of $\mathcal A$. Let $\delta_{ij}$ be a sequence of independent Bernoulli random variables such that
$$
\delta_{ij} = \left \{
\begin{array}{ll}
1, & \mathrm{if} \: (i,j) \in \Omega, \\
0, & \mathrm{otherwise.}
\end{array}
\right .
$$
Then, we can rewrite $\mathcal A$ as
$$
\mathcal A = \sum_{ij} \mathcal A_{ij},
$$
where
$$
\mathcal A_{ij} = \delta_{ij} \PQp(\bar{e}_i e_j^*) \otimes \PQp(\bar{e}_i e_j^*) - \frac{\rho}{mn}\PQp,
$$
and $\otimes$ denotes the outer or tensor product between matrices.
Then, we have that
\begin{align}
\|\mathcal A_{ij}\| & \le \| \PQp(\bar{e}_i e_j^*) \otimes \PQp(\bar{e}_i e_j^*) \| + \frac{\rho}{mn}\\
                    \label{eq:tensor}
                    & \le \| \PQp(\bar{e}_i e_j^*)\|_F^2 +\frac{\rho}{mn}\\
                    & \le \frac{\nu p}{n} + +\frac{\rho}{mn} \\
                    & \triangleq S,
\end{align}
where in Eqn. (\ref{eq:tensor}) we used the fact that $\|A \otimes B\| \le \|A\|_F \|B\|_F$.\par

We now bound the variance terms.
\begin{align}
\sigma^2 & =  \left\| \sum_{i,j} \E\left[\mathcal A_{ij}^2 \right] \right\|\\
        \label{eq:sigmaA}
         & =  \left\| \sum_{i,j} \left (\rho \left[\PQp(\bar{e}_i e_j^*) \otimes \PQp(\bar{e}_i e_j^*)\right]^2-\frac{2\rho^2 \PQp(\bar{e}_i e_j^*) \otimes \PQp(\bar{e}_i e_j^*)}{mn}+\frac{\rho^2}{m^2 n^2} \PQp\right ) \right\|
\end{align}
We let $\POij$ denote the orthogonal projector onto the subspace $\mathrm{span}({\bar{e}_i e_j^*})$. Clearly, we have
$$
\PO = \sum_{(i,j)\in \Omega} \POij.
$$
Furthermore, we note that
$$
\PQp(\bar{e}_i e_j^*) \otimes \PQp(\bar{e}_i e_j^*) = \PQp \POij \PQp.
$$
Thus, we get
\begin{eqnarray*}
\sum_{i,j} \frac{2\rho^2 \PQp(\bar{e}_i e_j^*) \otimes \PQp(\bar{e}_i e_j^*)}{mn}
& = & \frac{2\rho^2 \PQp \left( \sum_{i,j} \POij\right) \PQp}{mn}\\
& = & \frac{2\rho^2 \PQp}{mn}.
\end{eqnarray*}
Similarly, we have
\begin{eqnarray*}
\sum_{i,j} \rho \left[\PQp(\bar{e}_i e_j^*) \otimes \PQp(\bar{e}_i e_j^*)\right]^2
& = & \rho \sum_{i,j} \PQp \POij \PQp \POij \PQp\\
& = & \rho \PQp \left( \sum_{i,j} \POij \PQp \POij \right) \PQp.
\end{eqnarray*}
Let $X\in \R^{m\times n}$ be any matrix satisfying $\|X\|_F = 1$. Then,
\begin{eqnarray*}
\left\| \sum_{i,j} \POij \PQp \POij X \right\|_F & = & \left\| \sum_{i,j} \POij \left(\sum_{k=1}^p \braket{G_k, \POij X} G_k\right ) \right\|_F\\
                                                & = & \left\| \sum_{i,j} \POij \left (\sum_{k=1}^p \braket{\POij G_k, X} G_k\right) \right\|_F,
\end{eqnarray*}
where we recall that the $G_i$'s constitute an orthonormal basis for $Q^\perp$ satisfying $\max_i \|G_i\|^2 < \nu/n$.
We now bound $\|\POij G_k \|_F$ as follows:
\begin{eqnarray*}
\|\POij G_k \|_F & = & |\braket{\bar{e}_i e_j^*, \PQp G_k}| \\
& = & |\braket{\PQp \bar{e}_i e_j^*, G_k}| \\
& \le & \|G_k\| \|\PQp \bar{e}_i e_j^*\|_*\\
& \le & \|G_k\| \sqrt{n}\,  \|\PQp \bar{e}_i e_j^*\|_F\\
&\le & \sqrt{\frac{\nu}{n}}\: \sqrt{n}\: \sqrt{\frac{\nu p}{n}}\\
                        & = & \nu \sqrt{\frac{p}{n}}.
\end{eqnarray*}
Combining the above bound with H\"{o}lder's inequality, we get
\begin{align}
                                                \label{eq:expandPOij}
\left\| \sum_{i,j} \POij \PQp \POij(X) \right\|_F & \le \left\| \sum_{i,j,k } \braket{\POij G_k, X} \POij G_k \right\|_F\\
& \le \nu \sqrt{\frac{p}{n}}\left\| \left(\sum_{i,j} \POij \right)\left(\sum_{k=1}^p G_k\right)\right\|_F\\
                                                & \le \nu  \sqrt{\frac{p^3}{n}}.
\end{align}

Therefore, we have that the variance in Eqn. \eqref{eq:sigmaA} can be bounded as
\begin{eqnarray*}
\sigma^2 & \le &\rho  \nu  \sqrt{\frac{p^3}{n}}+ \frac{2\rho^2}{mn}+ \frac{\rho^2}{mn} \\
         & \le &2\rho  \nu  \sqrt{\frac{p^3}{n}}
\end{eqnarray*}
Applying the matrix Bernstein inequality (Theorem \ref{thm:operatorbound}), we get
\begin{eqnarray*}
\Pr\left[\|\mathcal A\| > t \right] &\le&  2m^2 \exp\left(-\frac{t^2}{2\sigma^2+ 3St}\right)\\
                                      &\le& 2m^2 \exp\left(-\frac{t^2}{C_1\rho \nu \sqrt{p^3/m}+ C_2 \nu pt/m}\right).
\end{eqnarray*}
Let us set $t=\rho$. Now, suppose that
\begin{align}
\frac{\nu^2 p^3 \log m}{n} \le C_3 \rho^2,
\end{align}
where $C_3 > 0$ is a numerical constant. Then, under the conditions of Theorem \ref{thm:cpcp_deterministic}, $\|\mathcal A\|$ is bounded from above by $\rho$ with high probability. Since $\mathcal A = \PQp \PO \PQp - \rho \PQp$, this implies that
\begin{align}
\|\PQp \PO \PQp \|\le 2\rho,
\end{align}
with high probability. It follows that
\begin{align}
\|\PQp \PO\| < 1/2,
\end{align}
with high probability, provided that $\rho$ is sufficiently small and $\nu^2 p^3 \log m/ n \le C$, where $C$ is a numerical constant.

{
\bibliographystyle{IEEEtran}
\bibliography{cpcp_ls_golfing}

\begin{thebibliography}{10}
\providecommand{\url}[1]{#1}
\csname url@samestyle\endcsname
\providecommand{\newblock}{\relax}
\providecommand{\bibinfo}[2]{#2}
\providecommand{\BIBentrySTDinterwordspacing}{\spaceskip=0pt\relax}
\providecommand{\BIBentryALTinterwordstretchfactor}{4}
\providecommand{\BIBentryALTinterwordspacing}{\spaceskip=\fontdimen2\font plus
\BIBentryALTinterwordstretchfactor\fontdimen3\font minus
  \fontdimen4\font\relax}
\providecommand{\BIBforeignlanguage}[2]{{%
\expandafter\ifx\csname l@#1\endcsname\relax
\typeout{** WARNING: IEEEtran.bst: No hyphenation pattern has been}%
\typeout{** loaded for the language `#1'. Using the pattern for}%
\typeout{** the default language instead.}%
\else
\language=\csname l@#1\endcsname
\fi
#2}}
\providecommand{\BIBdecl}{\relax}
\BIBdecl

\bibitem{Wright2009-PAMI}
J.~Wright, A.~Yang, A.~Ganesh, Y.~Ma, and S.~Sastry, ``Robust face recognition
  via sparse representation,'' \emph{{IEEE} Trans. Pattern Analysis and Machine
  Intelligence}, vol.~31, no.~2, Feb 2009.

\bibitem{Fazel2004-ACC}
M.~Fazel, H.~Hindi, and S.~Boyd, ``Rank minimization and applications in system
  theory,'' in \emph{American Control Conference}, June 2004.

\bibitem{Papadimitriou2000-JCSS}
C.~Papadimitriou, P.~Raghavan, H.~Tamaki, and S.~Vempala, ``Latent semantic
  indexing: A probabilistic analysis,'' \emph{Journal of Computer and System
  Sciences}, vol.~61, no.~2, Oct 2000.

\bibitem{Eckart1936-Psychometrika}
C.~Eckart and G.~Young, ``The approximation of one matrix by another of lower
  rank,'' \emph{Psychometrika}, vol.~1, pp. 211--218, 1936.

\bibitem{Jolliffe1986}
I.~Jolliffe, \emph{Principal Component Analysis}.\hskip 1em plus 0.5em minus
  0.4em\relax Springer-Verlag, 1986.

\bibitem{Candes2011-JACM}
E.~Cand{\`{e}s}, X.~Li, Y.~Ma, and J.~Wright, ``Robust principal component
  analysis?'' \emph{Journal of the ACM}, vol.~58, no.~7, May 2011.

\bibitem{Chandrasekaran2011-SJO}
V.~Chandrasekaran, S.~Sanghavi, P.~Parrilo, and A.~Willsky, ``Rank-sparsity
  incoherence for matrix decomposition,'' \emph{{SIAM} Journal on
  Optimization}, vol.~21, no.~2, pp. 572--596, 2011.

\bibitem{Hsu2011-IT}
D.~Hsu, S.~M. Kakade, and T.~Zhang, ``Robust matrix decomposition with sparse
  corruptions,'' \emph{{IEEE} Transactions on Information Theory}, vol.~57,
  no.~11, pp. 7221--7234, 2011.

\bibitem{Zhou2010-ISIT}
Z.~Zhou, X.~Li, J.~Wright, E.~Cand{\`{e}s}, and Y.~Ma, ``Dense error correction
  for low-rank matrices via principal component pursuit,'' in \emph{{IEEE}
  International Symposium on Information Theory}, 2010.

\bibitem{Ganesh2010-ISIT}
A.~Ganesh, J.~Wright, X.~Li, E.~Cand{\`{e}s}, and Y.~Ma, ``Dense error
  correction for low-rank matrices via principal component pursuit,'' in
  \emph{{IEEE} International Symposium on Information Theory}, 2010.

\bibitem{Wu2010-ACCV}
L.~Wu, A.~Ganesh, B.~Shi, Y.~Matsushita, Y.~Wang, and Y.~Ma, ``Robust
  photometric stereo via low-rank matrix completion and recovery,'' in
  \emph{Asian Conference on Computer Vision}, 2010.

\bibitem{Zhang2011-IJCV}
\BIBentryALTinterwordspacing
Z.~Zhang, A.~Ganesh, X.~Liang, and Y.~Ma, ``{TILT}: {T}ransform {I}nvariant
  {L}ow-rank {T}extures,'' \emph{International Journal of Computer Vision},
  2011. [Online]. Available: \url{http://dx.doi.org/10.1007/s11263-012-0515-x}
\BIBentrySTDinterwordspacing

\bibitem{Peng2011-PAMI}
Y.~Peng, A.~Ganesh, J.~Wright, W.~Xu, and Y.~Ma, ``{RASL}: {R}obust {A}lignment
  by {S}parse and {L}ow-rank decomposition,'' \emph{To appear in {IEEE}
  Transactions on Pattern Analysis and Machine Intelligence}, 2011.

\bibitem{Candes2008}
E.~Cand{\`{e}s} and B.~Recht, ``Exact matrix completion via convex
  optimzation,'' \emph{Found. of Comput. Math.}, 2008.

\bibitem{Candes2009}
E.~Cand{\`{e}s} and T.~Tao, ``The power of convex relaxation: Near-optimal
  matrix completion,'' \emph{{\rm to appear in} {IEEE} Transactions on
  Information Theory}, 2009.

\bibitem{Gross2009-pp}
D.~Gross, ``Recovering low-rank matrices from few coefficients in any basis,''
  \emph{{\rm preprint}}, 2009.

\bibitem{Li2011-pp}
X.~Li, ``Compressed sensing and matrix completion with constant proportion of
  corruptions,'' 2011, available at \url{http://arxiv.org/abs/1104.1041}.

\bibitem{Wright}
J.~Wright, A.~Ganesh, K.~Min, and Y.~Ma, ``Compressive principal component
  pursuit,'' 2012, submitted.

\bibitem{Ledoux}
M.~Ledoux, \emph{The Concentration of Measure Phenomenon}.\hskip 1em plus 0.5em
  minus 0.4em\relax American Mathematical Society, 2001.

\bibitem{Vershynin2010}
M.~Rudelson and R.~Vershynin, ``Non-asymptotic theory of random matrices:
  extreme singular values,'' in \emph{Proc. of International Congress of
  Mathematicians}, 2010.

\bibitem{Candes2007-IP}
E.~Cand{\`{e}s} and J.~Romberg, ``Sparsity and incoherence in compressive
  sampling,'' \emph{Inverse Problems}, vol.~23, no.~3, pp. 969--985, 2007.

\bibitem{Tropp2011-FCM}
\BIBentryALTinterwordspacing
J.~Tropp, ``User-friendly tail bounds for sums of random matrices,''
  \emph{Foundations of Computational Mathematics}, 2011. [Online]. Available:
  \url{http://dx.doi.org/10.1007/s10208-011-9099-z}
\BIBentrySTDinterwordspacing

\bibitem{Vershynin2011}
R.~Vershynin, ``Introduction to the non-asymptotic analysis of random
  matrices,'' 2011, available at
  \texttt{http://www-personal.umich.edu/~romanv/papers/non-asymptotic-rmt-plain.pdf}.

\end{thebibliography}
}

\appendix

\end{document}